\documentclass{jpc} 


\usepackage{hyperref}
\usepackage{natbib}
\usepackage[ruled]{algorithm2e}
\usepackage{algorithmic}
\usepackage{booktabs} 

\newcommand{\R}{\mathbb{R}} 

\newcommand{\cD}{\mathcal D}

\newcommand{\cI}{\mathcal I}

\newcommand{\cN}{\mathcal N}

\newcommand{\supp}{\mathrm{supp}} 

\newcommand{\norm}[1]{\left\lVert #1\right\rVert} 
\newcommand{\ellnorm}[2]{\norm{#1}_{\ell_{#2}}} 

\newcommand{\set}[1]{\{#1\}}

\renewcommand{\epsilon}{\varepsilon}
\newcommand{\eps}{\epsilon}

  \newcommand{\beq}{\begin{equation}}
  \newcommand{\eeq}{\end{equation}}
  \newcommand{\beqn}{\begin{equation*}}
  \newcommand{\eeqn}{\end{equation*}}
  \newcommand{\beqr}{\begin{eqnarray}}
  \newcommand{\eeqr}{\end{eqnarray}}
  \newcommand{\beqrn}{\begin{eqnarray*}}
  \newcommand{\eeqrn}{\end{eqnarray*}}
  \newcommand{\bmline}{\begin{multline}}
  \newcommand{\emline}{\end{multline}}
  \newcommand{\bmlinen}{\begin{multline*}}
  \newcommand{\emlinen}{\end{multline*}}

\theoremstyle{plain}
\theoremstyle{definition}
\newtheorem{definition}{Definition}[section]
\newtheorem{theorem}{Theorem}[section]
\newtheorem{proposition}{Proposition}[section]
\newtheorem{claim}{Claim}[section]
\newtheorem{lemma}{Lemma}[section]

\newtheorem{problem}{Problem}[section]

\newcommand{\COMMENTALG}[2][.7\linewidth]{\leavevmode\hfill\makebox[#1][l]{//~#2}}

\newcommand{\Lap}{\mathsf{Lap}}
\newcommand{\Gauss}{\mathsf{Gauss}}

\newcommand{\bigzero}{\mbox{\normalfont\Large\bfseries 0}}

\theoremstyle{plain} 

\begin{document}

\title[]{Differentially Private Set Union\rsuper*}
\titlecomment{{\lsuper*}A preliminary version of this paper appeared in ICML 2020 (\cite{gopi2020differentially}).}

\author[S.~Gopi]{Sivakanth Gopi}	
\address{Microsoft, Redmond, WA, USA}	
\email{sigopi@microsoft.com}  

\author[P.~Gulhane]{Pankaj Gulhane}	
\address{Microsoft, Redmond, WA, USA}	
\email{pagulhan@microsoft.com}  

\author[J.~Kulkarni]{Janardhan Kulkarni}	
\address{Microsoft, Redmond, WA, USA}	
\email{jakul@microsoft.com}  

\author[JH.~Shen]{Judy Hanwen Shen}	
\address{Computer Science Department, Stanford University, Palo Alto, CA, USA}	
\email{jhshen@cs.stanford.edu}  

\author[M.~Shokouhi]{Milad Shokouhi}	
\address{Microsoft, Redmond, WA, USA}	
\email{milads@microsoft.com}

\author[S.~Yekhanin]{Sergey Yekhanin}	
\address{Microsoft, Redmond, WA, USA}	
\email{yekhanin@microsoft.com}




\begin{abstract}
  \noindent We study the basic operation of set union in the global model of differential privacy. In this problem, we are given a universe $U$ of items, possibly of infinite size, and a database $D$ of users. Each user $i$ contributes a subset $W_i \subseteq U$ of items. We want an ($\epsilon$,$\delta$)-differentially private algorithm which outputs a subset $S \subset \cup_i W_i$ such that the size of $S$ is as large as possible. The problem arises in countless real world applications; it is particularly ubiquitous in natural language processing (NLP) applications as vocabulary extraction. For example, discovering words, sentences, $n$-grams etc., from private text data belonging to users is an instance of the set union problem. Known algorithms for this problem proceed by collecting a subset of items from each user, taking the union of such subsets, and disclosing the items whose noisy counts fall above a certain threshold. Crucially, in the above process, the contribution of each individual user is always independent of the items held by other users, resulting in a wasteful aggregation process, where some item counts happen to be way above the threshold. We deviate from the above paradigm by allowing users to contribute their items in a {\em dependent fashion}, guided by a {\em policy}. In this new setting ensuring privacy is significantly delicate. We prove that any policy which has certain {\em contractive} properties would result in a differentially private algorithm. We design two new algorithms for differentially private set union, one using Laplace noise and other Gaussian noise, which use $\ell_1$-contractive and $\ell_2$-contractive policies respectively and provide concrete examples of such policies. Our experiments show that the new algorithms in combination with our policies significantly outperform previously known mechanisms for the problem.
\end{abstract}

\maketitle

\section{Introduction}

\label{sec:intro}
Natural language models for applications such as suggested replies for e-mails and dialog systems rely on the discovery of $n$-grams and sentences \cite{HLLC14, KannanK16, ChenL19, DebB19}. Words and phrases used for training come from individuals, who may be left vulnerable if personal information is revealed. For example, a model could generate a sentence or predict a word that can potentially reveal personal information of the users in the training set \cite{CarliniL19}. 
Therefore, algorithms that allow the public release of the words, $n$-grams, and sentences obtained from users' text while preserving privacy are desirable. Additional applications of this problem include the release of search queries and keys in SQL queries \cite{KKMN09, WilsonZ20}.  
While other privacy definitions are common in practice, guaranteeing differential privacy, introduced in the seminal work of Dwork {\em et al} \cite{DMNS06}, ensures users the strongest preservation of privacy. In this paper we consider user level privacy.
    
\begin{definition}[Differential Privacy \cite{DR14}]
	A randomized algorithm $\mathcal{A}$ is  ($\epsilon$,$\delta$)-differentially private if for any two neighboring databases $D$ and $D'$, where a single user's data is removed from one database to obtain the other, and for all sets $\mathcal{S}$ of possible outputs: 
$$
\textstyle{\Pr[\mathcal{A}(D) \in \mathcal{S}] \leq e^{\epsilon}\Pr[\mathcal{A}(D') \in \mathcal{S}] +\delta.}
$$
\end{definition}

An algorithm satisfying differential privacy (DP) guarantees that its output does not change by much if a single user is either added or removed from the dataset. Moreover, the guarantee holds regardless of how the output of the algorithm is used downstream. Therefore, items (e.g. n-grams) produced using a DP algorithm can be used in other applications without any privacy concerns. Since its introduction a decade ago \cite{DMNS06}, differential privacy has become the de facto notion of privacy in statistical analysis and machine learning, with a vast body of research work (see Dwork and Roth \cite{DR14}
and Vadhan \cite{V17} for surveys) and growing acceptance in industry. Differential privacy is deployed in many industries, including Apple~\cite{A17}, Google~\cite{EPK14,BEMMR+17}, Microsoft~\cite{DKY17}, Mozilla~\cite{AKZHL17}, and the US Census Bureau~\cite{A16, KCKHM18}.

\smallskip
The vocabulary extraction and $n$-gram discovery problems mentioned above, as well as many commonly studied problems \cite{KKMN09, WilsonZ20}, can be abstracted as a set union which leads to the following problem.

\begin{problem}[Differentially Private Set Union (DPSU)]
	Let $U$ be some universe of items, possibly of unbounded size. Suppose we are given a database $D$ of users where each user $i$ has a subset $W_i \subseteq U$. We want an ($\epsilon$,$\delta$)-differentially private Algorithm $A$ which outputs a subset $S\subseteq \cup_i W_i$ such that the size of $S$ is as large as possible.
\end{problem}

Since the universe of items can be unbounded, as in our motivating examples, it is not clear how to apply the exponential mechanism \cite{McSherryT07} to DPSU.
Furthermore, even for the cases when $U$ is bounded, implementing the exponential mechanism can be also very inefficient.
Existing algorithms \footnote{They don't study the DPSU problem as defined in this paper. Their goal is to output approximate counts of as many items as possible in $\cup_i W_i.$} for this problem \cite{KKMN09, WilsonZ20} collect a bounded number of items from each user, build a histogram of these items, and disclose the items whose noisy counts fall above a certain threshold. In these algorithms, the contribution of each user is always {\em independent} from the identity of items held by other users, resulting in a wasteful aggregation process, where some items' counts could be far above the threshold. Since the goal is to release as large a set as possible rather than to release accurate counts of each item, there could be more efficient ways to allocate the weight to users' items. 

\begin{figure}[ht]
\begin{center}
\includegraphics[width=12cm]{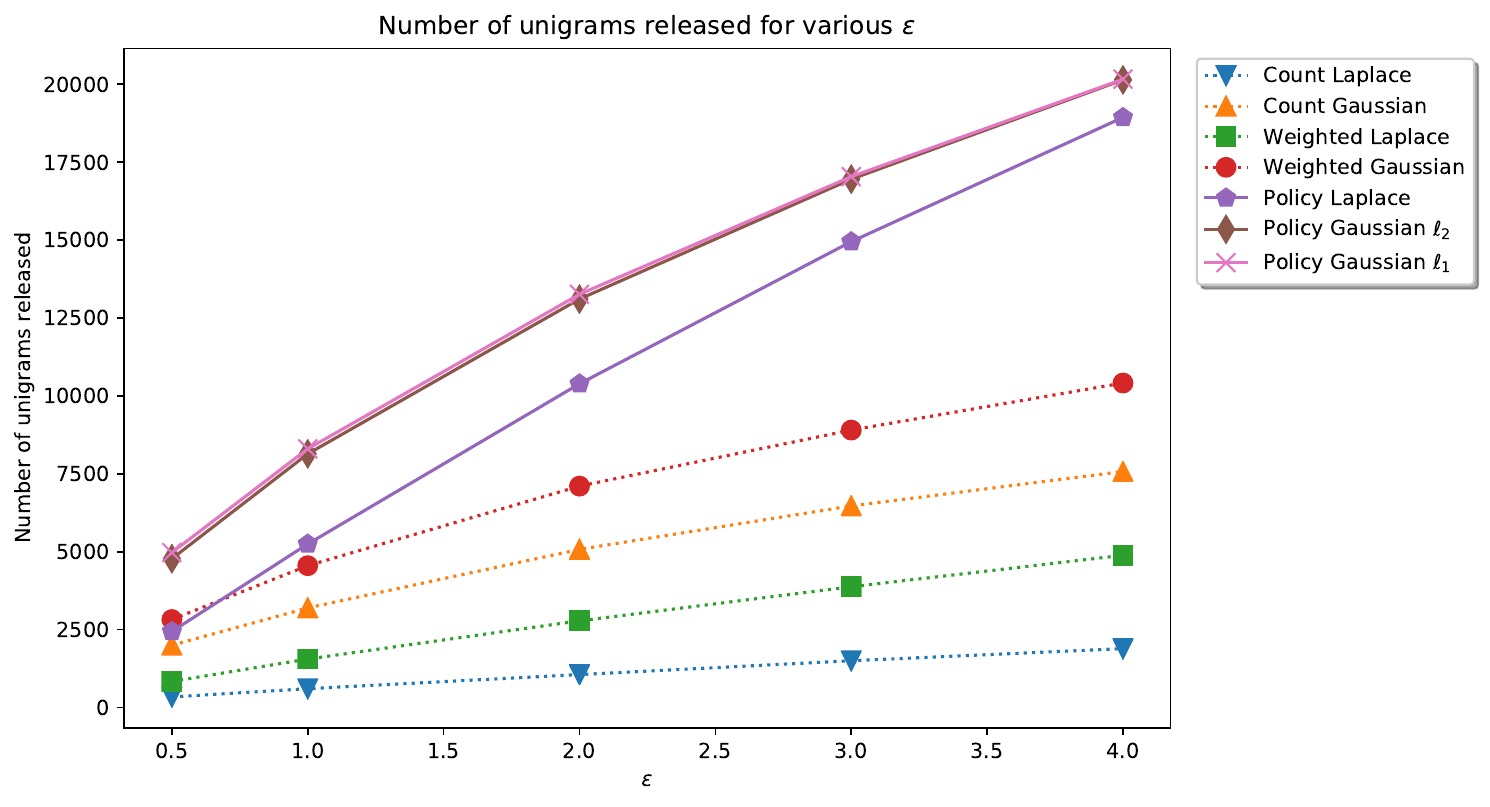}
\caption{\small Size of the set output by our proposed algorithms \textsc{Policy Laplace} and \textsc{Policy Gaussian} compared to natural generalizations of previously known algorithms for various values of privacy parameter $\eps$ and $\delta=\exp(-10)$.}
\label{fig:eps_fig1}
\end{center}
\end{figure}

\noindent We deviate from the previous methods by allowing users to contribute their items in a {\em dependent fashion}, guided by an {\em update policy}.
In our algorithms, proving privacy is more delicate as some update policies can result in histograms with unbounded sensitivity. 
We prove a meta-theorem to show that update policies with certain {\em contractive properties} would result in differentially private algorithms.
The main contributions of the paper are:
 
\begin{itemize}
\item Guided by our meta-theorems, we introduce two new algorithms called \textsc{Policy Laplace} and \textsc{Policy Gaussian} for the DPSU problem. Both of them run in \emph{linear time} and only require a single pass over the users' data.
\item Using a Reddit dataset, we demonstrate that our algorithms significantly improve the size of DP set union even when compared to natural generalizations of the existing mechanisms for this problem (see Figure \ref{fig:eps_fig1}). We also show that our algorithms compare favorably to \textsc{$k$-anonymity} which is an ad hoc method used in practice that is not differentially private.
\end{itemize}

\subsection{Baseline algorithms}
\label{sec:baseline}

To understand the DPSU problem better, let us start with the simplest case we can solve by known techniques.
Define $\Delta_0 = \max_{i} |W_i|$.
Suppose $\Delta_0 = 1$. This special case can be solved using the algorithms in \cite{KKMN09, WilsonZ20}.
Their algorithm works as follows: Construct a histogram on $\cup_{i}W_i$ (the set of items in a database $D$) where the count of each item is the number of sets it belongs to. Then add Laplace noise or Gaussian noise to the counts of each item.
Finally, release only those items whose noisy histogram counts are above a certain {\em threshold} $\rho$.
It is not hard to prove that if the threshold is set sufficiently high, then the algorithm is $(\epsilon, \delta)$-DP.

\smallskip

A straight-forward extension of the histogram algorithm for $\Delta_0 > 1$ is to upper bound the $\ell_1$-sensitivity by $\Delta_0$ (and $\ell_2$-sensitivity by $\sqrt{\Delta_0}$), and then add some appropriate amount of Laplace noise (or Gaussian noise) based on sensitivity. The threshold $\rho$ has to be set based on $\Delta_0.$ The Laplace noise based algorithm was also the approach considered in \cite{KKMN09, WilsonZ20}.
This approach has the following drawback. Suppose a significant fraction of users have sets of size smaller than $\Delta_0$. Then constructing a histogram based on {\em counts} of the items results in {\em wastage of sensitivity budget}. A user $i$ with $|W_i| < \Delta_0$ can increment the count of items in $W_i$ by any vector $v \in \R^{W_i}$ as long as one can ensure that $\ell_1$ sensitivity is bounded by $\Delta_0$ (or $\ell_2$ sensitivity is bounded by $\sqrt{\Delta_0}$ if adding Gaussian noise). Consider the following natural generalization of Laplace and Gaussian mechanisms  to create a {\em weighted histogram of elements}. A weighted histogram over a domain $X$ is any map $H:X\to \R$. For an item $u\in U,$ $H(u)$ is called the weight of $u.$ In the rest of the paper, the term histogram should be interpreted as weighted histogram. Each user $i$ updates the weight of each item $u \in W_i$ using the rule:$H[u] := H[u] + (\Delta_0/|W_i|)^{1/p}$ for $p =1$ or $p=2$.
It is not hard to see that $\ell_p$-sensitivity of this weighted histogram is still $\Delta_0^{1/p}$.
Adding Laplace noise (for $p = 1$) or Gaussian noise (for $p=2$) to each item of the weighted histogram, and releasing only those items above an appropriately calibrated threshold will lead to differentially private output.
We call these algorithms as \textsc{Weighted Laplace} and \textsc{Weighted Gaussian}, they will be used as benchmarks to compare against our new algorithms.

\textbf{Related Work:} After a preliminary version of our work was published (\cite{gopi2020differentially}), a simple and nearly optimal algorithm for DPSU in the special case where every user contributes exactly one item (i.e., $\Delta_0=1$) is given by \cite{desfontaines2020differentially}. Our DPSU algorithms have been used for differentially private $n$-gram extraction (DPNE) by~\cite{kim2021differentially}. In DPNE, the goal is to learn as many $n$-grams as possible of varying lengths from a corpus of text data, this can be thought of as a generalization of DPSU where we only learn $1$-grams.

\subsection{Our techniques}
\label{sec:techniques}

The \textsc{Weighted Laplace} and \textsc{Weighted Gaussian} mechanisms described above can be thought of trying to solve the following variant of a {\em Knapsack} problem. Here each item $u \in U$ is a bin and we gain a profit of 1 if the total weight of the item in the weighted histogram constructed is more than the threshold.  
Each user can increment the weight of elements $u \in W_i$ using an \emph{update policy} $\phi$ which is defined as follows.
\begin{definition}[Update policy]
   \label{def:update_policy}
    An update policy is a map $\phi:\R^U \times 2^U \to \R^U$ such that $\supp(\phi(H,W)-H)\subset W$, i.e., $\phi$ can only update the weights of items in $W$. And the $i^{th}$ user updates $H$ to $\phi(H,W_i).$ Since $W_i$ is typically understood from context, we will write $\phi(H)$ instead of $\phi(H,W_i)$ for simplicity.
\end{definition}

\smallskip
\noindent In this framework, the main technical challenge is the following:

\smallskip

{\em How to design update policies such that the sensitivity of the resulting weighted histogram is small while maximizing the number of bins that are full?}

\smallskip

Note that bounding sensitivity requires that $\ellnorm{\phi(H,W)-H}{p}\le C$ for some constant $C$ i.e. each user has an $\ell_p$-budget of $C$ and can increase the weights of items in their set by an $\ell_p$-distance of at most $C$. By scaling, WLOG we can assume that $C=1.$
Note that having a larger value of $\Delta_0$ should help in filling more bins as users have more choice in how they can use their budget to increment the weight of items.

In this paper, we consider algorithms which {\em iteratively} construct the weighted histogram. 
That is, in our algorithms, we consider users in a random order, and each user updates the weighted histogram using the update policy $\phi.$ Algorithm~\ref{alg:meta} is a meta-algorithm for DP set union, and all our subsequent algorithms follow this framework.

\begin{algorithm}[!ht]
 \caption{High level meta algorithm for DP Set Union}
 \label{alg:meta}
 
 \begin{algorithmic}
   \STATE {\bfseries Input:} $D$: Database of $n$ users where each user $i$ has some subset $W_i\subset U$\\
 $\rho$: threshold\\
 \textsf{Noise}: Noise distribution ($\Lap(0,\lambda)$ or $\cN(0,\sigma^2)$)\\
   \STATE {\bfseries Output:} S: A subset of $\cup_i W_i$ \\
   Build weighted histogram $H$ supported over $\cup_i W_i$ using Algorithm~\ref{alg:meta_histogram}.\\
   $S = \{\} $ (empty set)\\
   \FOR{$u \in \cup_i W_i$}
   \STATE $\hat{H}[u]\leftarrow H[u]+\textsf{Noise}$\\
   \IF{$\hat{H}[u]>\rho$}
   \STATE $S\leftarrow S\cup \{u\}$
   \ENDIF
   \ENDFOR
   \STATE Output $S$
\end{algorithmic}
\end{algorithm}

\begin{algorithm}[!ht]
 \caption{High level meta algorithm for building weighted histogram using a given update policy}
 \label{alg:meta_histogram}
 
 \begin{algorithmic}
   \STATE {\bfseries Input:} $D$: Database of $n$ users where each user $i$ has some subset $W_i\subset U$\\
 $\Delta_0$: maximum contribution parameter\\
 \textsf{hash}: A random hash function which maps user ids into some large domain without collisions\\
 $\phi$: Update policy for a user to update the weights of items in their set
   \STATE {\bfseries Output:} H: A weighted histogram in $\R^{\cup_i W_i}$ \\
 $H = \lbrace \rbrace$ (empty histogram) \\
   \STATE Sort users into $\mathrm{User}_1,\mathrm{User}_2,\dots,\mathrm{User}_n$ by sorting the \textsf{hash} values of their user ids\\
   \FOR{$i=1$ {\bfseries to} $n$}
   \STATE $W_i \leftarrow $ set with $\mathrm{User}_i$
   \IF{$|W_i|>\Delta_0$}
   \STATE $W_i' \leftarrow$ Randomly choose $\Delta_0$ items from $W_i$
   \ELSE 
   \STATE $W_i' \leftarrow W_i$
   \ENDIF
   \STATE Update $H[u]$ for each $u\in W_i'$ using update policy $\phi$\\
   \ENDFOR
   \STATE Output $H$
\end{algorithmic}
\end{algorithm}

If the update policy is such that it increments the weights of items independent of other users (as done in \textsc{Weighted Laplace} and \textsc{Weighted Gaussian}), then it is not hard to see that sensitivity of $H$ can be bounded by $1$; that is, by the budget of each user. However, if some item is already way above the threshold $\rho,$ then it does not make much sense to waste the limited budget on that item. Ideally, users can choose a clever \emph{update policy} to distribute their budget among the $W_i$ items based on the current weights. 

Note that if a policy is such that updates of a user depends on other users, it can be quite tricky to bound the sensitivity of the resulting weighted histogram. To illustrate this, consider for example the {\em greedy} update policy. Each user $i$ can use his budget of 1 to fill the bins that is closest to the threshold among the bins $u \in W_i$. If an item already reached the threshold, the user can spend his remaining budget incrementing the weight of next bin that is closest to the threshold and so on. Note that from our Knapsack problem analogy this seems be a good way to maximize the number of bins filled. However such a greedy policy can have very large sensitivity, and hence won't lead to any reasonable DP algorithm. 
So, the main contribution of the paper is in showing policies which help maximize the number of item bins that are filled while keeping the sensitivity low. 
In particular, we define a general class of $\ell_p$-contractive update policies and show that they produce weighted histograms with bounded $\ell_p$-sensitivity.

\begin{definition}[$\ell_p$-contractive update policy]
\label{def:ellp_contractive_policy}
   We say that an update policy $\phi$ is $\ell_p$-contractive if there exists a subset $\cI$ (called the invariant subset for $\phi$) of pairs of weighted histograms which are at an $\ell_p$ distance of at most 1, i.e., $$\cI\subset \left\{(H_1,H_2): \ellnorm{H_1-H_2}{p}\le 1\right\}$$ such that the following conditions hold.
   \begin{enumerate}
      \item (Invariance) $(H_1,H_2)\in \cI \Rightarrow (\phi(H_1,W),\phi(H_2,W))\in \cI$ for all $W$.\footnote{Note that property (1) is a slightly weaker requirement than the usual notion of $\ell_p$-contractivity which requires $\ellnorm{\phi(H_1,W)-\phi(H_2,W)}{p}\le \ellnorm{H_1-H_2}{p}$ for all $H_1,H_2.$ Instead we require contraction only for $(H_1,H_2)\in \cI.$ }
      \item $(\phi(H,W),H)\in \cI$ for all $H,W$.
   \end{enumerate}
\end{definition}
Property (2) of Definition~\ref{def:ellp_contractive_policy} requires that the update policy can change the histogram by an $\ell_p$ distance of at most 1 (budget of a user).

\begin{theorem}[Contractivity implies bounded sensitivity]
\label{thm:contractivity_implies_sensitivity}
Suppose $\phi$ is an update policy which is $\ell_p$-contractive over some invariant subset $\cI$. Then the histogram output by Algorithm~\ref{alg:meta_histogram} (for any fixed choice of $W_i'\subset W_i$ for each user) has $\ell_p$-sensitivity bounded by 1.
\end{theorem}
We prove Theorem~\ref{thm:contractivity_implies_sensitivity} in Section~\ref{sec:contractivity}.
Once we have bounded $\ell_p$-sensitivity, we can get a DP Set Union algorithm with some additional technical work as stated in this informal theorem (see Appendix~\ref{sec:sensitivity_implies_DP} for a formal version).
\begin{theorem}(Informal: Bounded sensitivity implies DP)
   \label{thm:sensitivity_implies_DP_informal}
   For $p\in \{1,2\}$, if the $\ell_p$-sensitivity of the weighted histogram output by Algorithm~\ref{alg:meta_histogram} is bounded, then Algorithm~\ref{alg:meta} for DP Set Union can be made $(\eps,\delta)$-differentially private by appropriately choosing the noise distribution (\textsf{Noise}) and threshold ($\rho$). 
\end{theorem}

The main contribution of the paper is two new algorithms and appropriate contractive update policies guided by Theorem \ref{thm:contractivity_implies_sensitivity}. The first algorithm, which we call \textsc{Policy Laplace}, uses policies which are $\ell_1$-contractive. The second algorithm, which we call \textsc{Policy Gaussian}, uses policies which are $\ell_2$-contractive. Finally we show that our algorithms with appropriate update policies significantly outperform the weighted update policies. 

At a very high-level, the role of contractivity in our algorithms is indeed similar to its role in the recent elegant work of Feldman {\em et al} \cite{FeldmanMTT18}. They show that if an iterative algorithm is contractive in each step, then adding Gaussian noise in each iteration will lead to strong privacy amplification. In particular, users who make updates early on will enjoy much better privacy guarantees. 
However their framework is not applicable in our setting, because their algorithm requires adding noise to the count of every item in every iteration; this will lead to unbounded growth of counts and items which belong to only a single user can also get output which violates privacy.

\section{Preliminaries}
\label{sec:prelims}
Let $\cD$ denote the collection of all databases. We say that $D,D'$ are neighboring databases, denoted by $D\sim D'$, if they differ in exactly one user.
\begin{definition}
For $p\ge 0,$ the $\ell_p$-sensitivity of $f:\cD\to \R^k$ is defined as $\sup_{D\sim D'} \ellnorm{f(D)-f(D')}{p}$ where the supremum is over all neighboring databases $D,D'$.
\end{definition}

\begin{proposition}[The Laplace Mechanism \cite{DR14}] Given any function $f: D \rightarrow \mathbb{R}^k$, the Laplace Mechanism is defined as: 
\begin{equation}
	\mathcal{M}(x, f(.), \eps) = f(x) + (Y_i, ..., Y_k)
\end{equation}
where $\Delta_1$ is the $\ell_1$-sensitivity and $Y_i$ are i.i.d. random variables drawn from $\Lap(0, \Delta_1/\eps)$ .
\end{proposition}

\begin{proposition}[Gaussian Mechanism~\cite{BalleW18}]
\label{lem:gaussian_mechanism}
		If $f:\cD \to \mathbb{R}^d$ is a function with $\ell_2$-sensitivity $\Delta_2$. For any $\eps\ge 0$ and $\delta\in [0,1]$, the Gaussian output perturbation mechanism $M(x)=f(x)+Z$ with $Z\sim \cN(0,\sigma^2 I)$ is $(\eps,\delta)$-DP if and only if $$\Phi\left(\frac{\Delta_2}{2\sigma}-\frac{\eps\sigma}{\Delta_2}\right)-e^\eps\Phi\left(-\frac{\Delta_2}{2\sigma}-\frac{\eps\sigma}{\Delta_2}\right)\le \delta.$$ 
\end{proposition}

\begin{definition}
	We say that two distributions $P,Q$ on a domain $\Omega$ are $(\eps,\delta)$-close to each other, denoted by $P\approx_{\eps,\delta} Q$, if for every $S\subset \Omega$, we have
	\begin{enumerate}
		\item  $\Pr_{x\sim P}[x\in S]\le e^\eps \Pr_{x\sim Q}[x\in S]+\delta$ and
		\item  $\Pr_{x\sim Q}[x\in S]\le e^\eps \Pr_{x\sim P}[x\in S]+\delta.$
	\end{enumerate}
	We say that two random variables $X,Y$ are $(\eps,\delta)$-close to each other, denoted by $X\approx_{\eps,\delta} Y$, if their distributions are $(\eps,\delta)$-close to each other.
\end{definition}

We will need the following lemmas which are useful to prove $(\eps,\delta)$-DP.

\begin{lemma}
	\label{lem:eps_delta_conditioning}
	Let $P,Q$ be probability distributions over a domain $X$. If there exists an event $E$ s.t. $P[E]=1-\delta'$ and $P|_E \approx_{\eps,\delta} Q$, then $P \approx_{\eps,\delta+\delta'}Q$.
\end{lemma}
\begin{proof}
	Fix some subset $S\subseteq X$.
	\begin{align*}
		\Pr_{x\sim P}[x\in S] &= P[\bar E] \Pr_{x\sim P}[x\in S|\bar E] + P[E] \Pr_{x\sim P}[x\in S|E]\\
		&\le P[\bar E]+\Pr_{x\sim P}[x\in S|E]\\
		&= \delta'+\Pr_{x\sim P|_E}[x\in S]\\
		&\le \delta'+e^\eps \Pr_{x\sim Q}[x\in S]+\delta
	\end{align*}
	We now prove the other direction.
	\begin{align*}
		\Pr_{x\sim Q}[x\in S] &\le e^\eps \Pr_{x\sim P|_E}[x\in S]+\delta\\
		&\le e^\eps \frac{\Pr_{x\sim P}[x\in S]}{P(E)}+\delta\\
		&= e^\eps \frac{\Pr_{x\sim P}[x\in S]}{1-\delta'}+\delta\\
		&= e^\eps\Pr_{x\sim P}[x\in S] +\delta'\left(\frac{e^\eps\Pr_{x\sim P}[x\in S]}{1-\delta'}\right)+\delta
	\end{align*}
	Now if $e^\eps\Pr_{x\sim P}[x\in S]\le 1-\delta'$, then we have $Pr_{x\sim Q}[x\in S] \le e^\eps\Pr_{x\sim P}[x\in S]+\delta'+\delta.$
	Otherwise, trivially $$\Pr_{x\sim Q}[x\in S]\le 1 \le e^\eps\Pr_{x\sim P}[x\in S]+\delta'+\delta.$$
\end{proof}

We will also need the fact that if $X\approx_{\eps,\delta} Y$, then after post-processing they also remain $(\eps,\delta)$-close.
\begin{lemma}[\cite{DR14}]
	\label{lem:post_processing_eps_delta}
	If two random variables $X,Y$ are $(\eps,\delta)$-close and $M$ is any randomized algorithm, then $M(X) \approx_{\eps,\delta} M(Y)$.
\end{lemma}

\section{Contractivity implies Bounded Sensitivity}
\label{sec:contractivity}

In this section, we prove Theorem~\ref{thm:contractivity_implies_sensitivity} which claims that if an update policy satisfies contractive property as in Definition \ref{def:ellp_contractive_policy}, then it implies bounded sensitivity of the histogram built by Algorithm~\ref{alg:meta_histogram}. This in turn implies a DPSU algorithm by Theorem~\ref{thm:sensitivity_implies_DP_informal}. 


\begin{proof}[Proof of Theorem~\ref{thm:contractivity_implies_sensitivity}]
Let $\phi$ be an $\ell_p$-contractive update policy with invariant subset $\cI.$
Consider two neighboring databases $D_1$ and $D_2$ where $D_1$ has one extra user compared to $D_2$. Let $H_1$ and $H_2$ denote the histograms built by Algorithm \ref{alg:meta} using the update policy $\phi$ when the databases are $D_1$ and $D_2$ respectively.

  Say the extra user in $D_1$ has position $t$ in the global ordering given by the hash function. Let $H_1^{t-1}$ and $H_2^{t-1}$ be the histograms after the first $t-1$ (according to the global order given by the hash function $\textsf{hash}$) users' data is added to the histogram. Therefore $H_1^{t-1}=H_2^{t-1}.$ And the new user updates $H_1^{t-1}$ to $H_1^{t}$. By property (2) in Definition~\ref{def:ellp_contractive_policy} of $\ell_p$-contractive policy, $(\phi(H_1^{t-1})),H_1^{t-1})\in \cI$. Since $\phi(H_1^{t-1})=H_1^t$, we have $(H_1^{t},H_1^{t-1})=(H_1^{t},H_2^{t-1})\in \cI.$ The remaining users are now added to $H_1^{t},H_2^{t-1}$ in the same order. Note that we are using the fact that the users are sorted according some hash function and they contribute in that order (this is also needed to claim that $H_1^{t-1}=H_2^{t-1}$). Therefore, by property (1) in Definition~\ref{def:ellp_contractive_policy} of $\ell_p$-contractive policy, we get $(H_1,H_2)\in \cI$. Since $\cI$ only contains pairs with $\ell_p$-distance at most 1, we have $\ellnorm{H_1-H_2}{p}\le 1$. Therefore the histogram built by Algorithm~\ref{alg:meta_histogram} using $\phi$ has $\ell_p$-sensitivity of at most 1.
\end{proof}
Above theorem implies that once we have a $\ell_p$ contractive update policy, we can appeal to Theorem \ref{thm:sensitivity_implies_DP_informal} to design an algorithm for DPSU.

\section{Policy Laplace algorithm} 
\label{sec:policylaplace}
 In this section we will present a DPSU algorithm called \textsc{Policy Laplace} which uses any symmetric $\ell_1$-contractive update policy. An update policy is called \emph{symmetric} if it updates items with equal weights by equal amounts. Later, in Section~\ref{sec:ell1-descent-ell1-contractivity}, we present a specific symmetric $\ell_1$-contractive update policy called $\ell_1$-descent (Algorithm~\ref{alg:ell1-descent-ell1-contractive}). We can also use contractive update policies which are not symmetric with a small increase in the threshold $\rho$, see Appendix~\ref{sec:sensitivity_implies_DP}.

 The \textsc{Policy Laplace} algorithm is described in Algorithm~\ref{alg:policy-laplace}. The cutoff parameter $\Gamma$ will be used in the update policy (Algorithm~\ref{alg:ell1-descent-ell1-contractive}). Intuitively, the update policy will stop increasing weights of items whose weights reach a cutoff $\Gamma.$ Since the added noise is $\Lap(0,\lambda)$, which is centered at 0, we want to set the cutoff $\Gamma$ in the update policy to be sufficiently above the threshold $\rho$. Thus we pick $\Gamma = \rho_{\Lap} + \alpha \cdot \lambda$ for some $\alpha>0$. From our experiments, choosing $\alpha\in [2,6]$ works best empirically. The parameters $\lambda,\rho_\Lap$ are set so as to achieve $(\epsilon,\delta)$-DP as shown in Theorem~\ref{thm:policy-laplace}. 

\begin{algorithm}[!h]
 \caption{\textsc{Policy Laplace} algorithm for DPSU}
 \label{alg:policy-laplace}
 \begin{algorithmic}
 \STATE {\bfseries Input:} $D$: Database of $n$ users where each user has some subset $W\subset U$\\
 $\Delta_0$: maximum contribution parameter\\
 $(\eps,\delta)$: privacy parameters\\ 
 $\alpha$: parameter for setting cutoff 
\STATE {\bfseries Output:} S: A subset of $\cup_i W_i$ 
\STATE $\lambda\leftarrow 1/\eps$ \COMMENTALG{Noise parameter in $\Lap(0,\lambda)$}
\STATE // Threshold parameter
\STATE $\rho_{\Lap} \leftarrow \max_{1\le t\le \Delta_0} \frac{1}{t}+\frac{1}{\epsilon}\log\left(\frac{1}{2\left(1-(1-\delta)^{1/t}\right)}\right)$
\STATE $\Gamma\leftarrow \rho_{\Lap}+ \alpha\cdot \lambda$ \COMMENTALG{Cutoff parameter for update policy}
\STATE Run Algorithm~\ref{alg:meta} with $\textsf{Noise}\sim \Lap(0,\lambda)$ and any symmetric $\ell_1$-contractive update policy (such as Algorithm~\ref{alg:ell1-descent-ell1-contractive} with cutoff parameter $\Gamma$) to output $S$.
\end{algorithmic}
\end{algorithm}

\subsection{Privacy analysis of \textsc{Policy Laplace}}
In this section, we will prove that the \textsc{Policy Laplace} algorithm (Algorithm~\ref{alg:policy-laplace}) satifies $(\eps,\delta)$-DP. By Theorem~\ref{thm:contractivity_implies_sensitivity} and Theorem~\ref{thm:sensitivity_implies_DP_informal}, we already have an intuitive path to prove privacy. 

We now state the privacy claims formally.
\begin{theorem}
\label{thm:policy-laplace}
The \textsc{Policy Laplace} algorithm (Algorithm \ref{alg:policy-laplace}) is $(\epsilon, \delta)$-$DP$ when $$\rho_{\Lap} \geq \max_{1\le t\le \Delta_0} \frac{1}{t}+\frac{1}{\epsilon}\log\left(\frac{1}{2\left(1-(1-\delta)^{1/t}\right)}\right).$$
\end{theorem}

 \begin{proof}
 Suppose $D_1$ and $D_2$ are neighboring databases where $D_1$ has one extra user compared to $D_2$. Let $P$ and $Q$ denote the distribution of output of the algorithm when the database is $D_1$ and $D_2$ respectively.
 We want to show that $P\approx_{\eps,\delta} Q$. It is enough to prove this for any fixed choice of $W_i'\subset W_i$ (in Algorithm~\ref{alg:meta_histogram}) identical in both instances, which corresponds to a coupling. Let $E$ be the event that the final output $A\subset \supp(H_2).$ 
 \begin{claim}
 	$P|_E \approx_{\eps,0} Q$
 \end{claim}
 \begin{proof}
 	Let $H_1$ and $H_2$ be the histograms generated by the algorithm from databases $D_1$ and $D_2$ respectively. And $\hat H_1$ and $\hat H_2$ be the histograms obtained by adding $\Lap(0,1/\eps)$ noise to each entry of $H_1$ and $H_2$ respectively. For any possible output $A$ of Algorithm~\ref{alg:policy-laplace}, we have
 	$$Q(A)=\Pr[A=\set{u\in \supp(H_2): \hat H_2[u]>\rho_{\Lap}}] \text{ and } P|_E(A)=\Pr[A=\set{u\in \supp(H_2): \hat H_1[u]>\rho_{\Lap}}].$$
 	So $A\sim P|_E$ is obtained by post-processing $\hat H_1|_E$ and $A\sim Q$ is obtained by post-processing $\hat H_2$.
 	Since post-processing only makes two distributions closer (Lemma~\ref{lem:post_processing_eps_delta}), it is enough to show that the distributions of the $\hat H_1|_{\supp(H_2)}$ and $\hat H_2$ are $(\eps,0)$-close to each other. By Theorem~\ref{thm:contractivity_implies_sensitivity}, $H_1|_{\supp(H_2)}$ and $H_2$ differ in $\ell_1$-distance by at most 1. Therefore $P|_E \approx_{\eps,0} Q$ by the properties of Laplace mechanism (see Theorem 3.6 in \cite{DR14}).
 \end{proof}
 By Lemma~\ref{lem:eps_delta_conditioning}, it is enough to show that $P(E)\ge 1-\delta$. Let $T=\supp(H_1)\setminus \supp(H_2).$ Note that $|T|\le \Delta_0$ and $H_1[u]\le \frac{1}{|T|}$ for $u\in T$ since the update policy is symmetric.

\begin{align}
    P(\bar E) &= \Pr[\exists u \in T\ |\ \hat{H_1}[u] > \rho_{\Lap}] \nonumber \\
    &=1-\Pr[\forall u\in T\ \ \hat{H_1}[u] \le \rho_{\Lap}] \nonumber \\
    &= 1-\prod_{u\in T}\Pr[H_1[u] + X_u \le \rho_{\Lap}] \tag{$X_u \sim \Lap(1/\eps)$} \nonumber \\
    &\le 1-\prod_{u\in T}\Pr\left[X_u \le \rho_{\Lap} - \frac{1}{|T|}\right]  \tag{$H_1[u]\le \frac{1}{|T|}$ for $u\in T$}\nonumber \\
    &= 1-\left(1-\frac{1}{2}\exp\left(-\eps \rho_\Lap + \eps \frac{1}{|T|}\right)\right)^{|T|} \label{eqn:union_bound_laplace}
\end{align}
 Thus for $$\rho_{\Lap} \geq \max_{1\le t\le \Delta_0} \frac{1}{t}+\frac{1}{\epsilon}\log\left(\frac{1}{2\left(1-(1-\delta)^{1/t}\right)}\right),$$ we have $P(\bar E)\le \delta$. Therefore the \textsc{Policy Laplace} algorithm (Algorithm \ref{alg:policy-laplace}) is $(\eps,\delta)$-DP.
 \end{proof}

 \subsection{\textsc{$\ell_1$-descent} update policy for $\ell_1$-contractivity} 
 \label{sec:ell1-descent-ell1-contractivity}
We will now describe a specific $\ell_1$-contractive policy called \textsc{$\ell_1$-descent}. The policy is described in Algorithm~\ref{alg:ell1-descent-ell1-contractive}.
 We will set some \emph{cutoff} $\Gamma$ above the threshold $\rho$ to use in the update policy. Once the weight of an item ($H[u]$) crosses the cutoff, we do not want to increase it further. In this policy, each user starts with a budget of 1. The user uniformly increases $H[u]$ for each $u\in W$ s.t. $H[u]<\Gamma$. Once some item's weight reaches $\Gamma,$ the user stops increasing that item and keeps increasing the rest of the items uniformly until the budget of 1 is expended.

This policy can also be interpreted as \emph{gradient descent} to minimize the $\ell_1$-distance between the current weighted histogram and the point $(\Gamma,\Gamma,\dots,\Gamma)$, hence the name \textsc{$\ell_1$-descent}. Since the gradient vector is 1 in coordinates where the weight is below cutoff $\Gamma$ and $0$ in coordinates where the weight is $\Gamma,$ the \textsc{$\ell_1$-descent} policy is moving in the direction of the gradient until it has moved a total $\ell_1$-distance of at most 1.

\begin{algorithm}[H]
 \caption{\textsc{$\ell_1$-descent} update policy for $\ell_1$-contractivity}
 \label{alg:ell1-descent-ell1-contractive}
 \begin{algorithmic}
 \STATE {\bfseries Input:} $H_0$: Current histogram \\
 $W$: A subset of $U$ of size at most $\Delta_0$\\
$\Gamma$: cutoff parameter
\STATE {\bfseries Output:} $H_1$: Updated histogram
\STATE
\STATE $H_1|_{U\setminus W} \leftarrow H_0|_{U\setminus W}$
\STATE $G\leftarrow (\Gamma,\Gamma,\dots,\Gamma)-H_0|_W$ 
\IF {$\norm{G}_{\ell_1} \le 1$}
	\STATE $H_1|_W \leftarrow H_0|_W$
\ELSE
	\STATE Find $\lambda \ge 0$ such that $\sum_{u\in W} \min\{G[u],\lambda\} = 1$
	\STATE $H_1[u] \leftarrow H_0[u] + \min\{G[u],\lambda\} \ \forall u\in W$
\ENDIF
\end{algorithmic}
\end{algorithm}

We will now prove that the \textsc{$\ell_1$-descent} policy in Algorithm~\ref{alg:ell1-descent-ell1-contractive} is indeed $\ell_1$-contractive. For two histograms $G_1,G_2$, we write $G_1 \ge G_2$ if $G_1[u] \ge G_2[u]$ for each every item $u$. $G_1\le G_2$ is defined similarly. 
\begin{lemma}
	\label{lem:contraction_ell1}
	Let $\cI=\{(G_1,G_2):G_1\ge G_2,\ \ellnorm{G_1-G_2}{1}\le 1\}$. Then \textsc{$\ell_1$-descent} update policy in Algorithm~\ref{alg:ell1-descent-ell1-contractive} is $\ell_1$-contractive over the invariant subset $\cI.$
\end{lemma}
\begin{proof}
	Let $\phi$ denote the \textsc{$\ell_1$-descent} update policy.

	We will first show property (2) of Definition~\ref{def:ellp_contractive_policy}. Let $G$ be any weighted histogram and let $G'=\phi(G)$. Clearly $G'\ge G$ as the new user will never decrease the weight of any item. Moreover, the total change to the histogram is at most $1$ in $\ell_1$-distance. Therefore $\ellnorm{G'-G}{1}\le 1.$ Therefore $(G',G)\in \cI.$

	We will now prove property (1) of Definition~\ref{def:ellp_contractive_policy}. Let $(G_1,G_2)\in \cI$, i.e., $G_1\ge G_2$ and $\ellnorm{G_1-G_2}{1}\le 1$. Let $G_1'=\phi(G_1), G_2'=\phi(G_2).$
	A new user can increase $G_1$ and $G_2$ by at most 1 in $\ell_1$ distance. Let $\Gamma$ be the cutoff parameter in Algorithm~\ref{alg:ell1-descent-ell1-contractive}. Let $S$ be the set of $\Delta_0$ items with the new user, therefore only the items in $S$ will change in $G_1',G_2'$. WLOG, we can assume that the user changes both $G_1$ and $G_2$ by exactly total $\ell_1$ distance of 1. Otherwise, in at least one of them all the items in $S$ should reach the cutoff $\Gamma$. If this happens with $G_1,$ then clearly $\Gamma=G_1'[u]\ge G_2'[u]$ for all $u\in S$. But it is easy to see that if this happens with $G_2$, then it should also happen with $G_1$ in which case $G_1'[u]=G_2'[u]=\Gamma$ for $u\in S.$ 

	Imagine that at time $t=0$, the user starts pushing mass continuously at a rate of 1 to both $G_1,G_2$ until the entire mass of $1$ is sent, which happens at time $t=1$. The mass flow is equally split among all the items which haven't yet crossed cutoff. Let $G_1^t$ and $G_2^t$ be the histograms at time $t\in [0,1]$ as mass is pushed constinuously at a constrant rate. Therefore, for $i=1,2$, $G_i^0=G_i$ and $G_i^1=G_i'$. We claim that $G_1^t \ge G_2^t$ implies that $\frac{d G_1^t[u]}{dt}\ge \frac{d G_2^t[u]}{dt}$ for all $u\in S$ s.t. $G_1^t[u]<\Gamma$. This is because the flow is split equally among items which didn't cross the cutoff, and there can only be more items in $G_2^t$ which didn't cross the the cutoff when compared to $G_1^t$. And at time $t=0$, we have $G_1^0\ge G_2^0$. Therefore, we have $G_1^t \ge G_2^t$ for all $t\in [0,1]$ and so $G_1'\ge G_2'.$

	We will now prove $\ell_1$-contraction. Let $C_i=\norm{G_i-G_i'}_{\ell_1}$. By the discussion above, $C_1\le C_2$ (either total mass flow is equal to 1 for both or all items in $S$ will reach cutoff $\Gamma$ in $G_1$ before this happens in $G_2$).
	\begin{align*}
		&\norm{G_1'-G_2'}_{\ell_1} \\
		&= \sum_{u \in S} G_1'[u] - \sum_{u\in S} G_2'[u] \tag{Since $G_1'\ge G_2'$}\\
		&= \sum_{u \in S} G_1[u] - \sum_{u\in S} G_2[u]  + C_1 - C_2\\
		&\le \sum_{u \in S} G_1[u] - \sum_{u\in S} G_2[u]  \tag{Since $C_1\le C_2$}\\
		&= \norm{G_1-G_2}_{\ell_1} \tag{Since $G_1 \ge G_2$}\\
		&\le 1.
	\end{align*}
	Therefore $(G_1',G_2')\in \cI$ which proves property (2) of Definition~\ref{def:ellp_contractive_policy}.
\end{proof}

\section{Policy Gaussian algorithm}
\label{sec:policyGuassian}

In this section we will present a DPSU algorithm called \textsc{Policy Gaussian} which uses any symmetric $\ell_2$-contractive update policy. An update policy is called \emph{symmetric} if it updates items with equal weights by equal amounts. Later, we will present two specific symmetric $\ell_2$-contractive update policies called $\ell_1$-descent (Algorithm~\ref{alg:ell1-descent-ell2-contractive}) and $\ell_2$-descent (Algorithm~\ref{alg:ell2-descent-ell2-contractive}).  We can also use contractive update policies which are not symmetric with a small increase in the threshold $\rho$, see Appendix~\ref{sec:sensitivity_implies_DP}.

The \textsc{Policy Gaussian} algorithm is described in Algorithm~\ref{alg:policy-gaussian}. The cutoff parameter $\Gamma$ will be used in the update policy (Algorithm~\ref{alg:ell2-descent-ell2-contractive} and \ref{alg:ell1-descent-ell2-contractive}). Intuitively, the update policy will stop increasing weights of items whose weights reach a cutoff $\Gamma.$ Since the added noise is $\cN(0,\sigma^2)$ which is centered at 0, we want to set the cutoff $\Gamma$ in the update policy to be sufficiently above (but not too high above) the threshold $\rho_\Gauss$. Thus we pick $\Gamma = \rho_{\Gauss} + \alpha \cdot \sigma$ for some $\alpha>0$. From our experiments, choosing $\alpha\in [2,6]$ empirically yields these best results. The parameters $\sigma,\rho_\Gauss$ are set so as to achieve $(\eps,\delta)$-DP as shown in Theorem~\ref{thm:policy-gaussian}. $\Phi(\cdot)$ is the cumulative density function of standard Gaussian distribution and $\Phi^{-1}(\cdot)$ is its inverse.

\begin{algorithm}[ht]
 \caption{\textsc{Policy Gaussian} algorithm for DPSU}
 \label{alg:policy-gaussian}
 \begin{algorithmic}
 \STATE {\bfseries Input:} $D$: Database of $n$ users where each user has some subset $W\subset U$\\
 $\Delta_0$: maximum contribution parameter\\
 $(\eps,\delta)$: privacy parameters\\
 $\alpha$: parameter for setting cutoff
\STATE {\bfseries Output:} S: A subset of $\cup_i W_i$
\STATE  // Standard deviation in Gaussian noise\\
\STATE $\sigma\leftarrow \min\left\{\sigma:\Phi\left(\frac{1}{2\sigma}-\eps\sigma\right)-e^\eps\Phi\left(-\frac{1}{2\sigma}-\eps\sigma\right)\le \frac{\delta}{2}\right\}$
\STATE  // Threshold parameter
\STATE $\rho_{\Gauss} \leftarrow \max_{1\le t \le \Delta_0}\left(\frac{1}{\sqrt{t}}+\sigma \Phi^{-1}\left(\left(1-\frac{\delta}{2}\right)^{1/t}\right)\right)$ 
\STATE $\Gamma \leftarrow \rho_{\Gauss}+ \alpha\cdot \sigma$ \COMMENTALG{Cutoff parameter for update policy}
\STATE Run Algorithm~\ref{alg:meta} with $\textsf{Noise}\sim \cN(0,\sigma^2)$ and any symmetric $\ell_2$-contractive update policy (such as Algorithm~\ref{alg:ell2-descent-ell2-contractive} or \ref{alg:ell1-descent-ell2-contractive} with cutoff parameter $\Gamma$) to output $S$.
	\end{algorithmic}
\end{algorithm}

To find $\min\left\{\sigma:\Phi\left(\frac{1}{2\sigma}-\eps\sigma\right)-e^\eps\Phi\left(-\frac{1}{2\sigma}-\eps\sigma\right)\le \frac{\delta}{2}\right\}$, one can use binary search because $\Phi\left(\frac{1}{2\sigma}-\eps\sigma\right)-e^\eps\Phi\left(-\frac{1}{2\sigma}-\eps\sigma\right)$ is a decreasing function of $\sigma.$ An efficient and robust implementation of this binary search can be found in~\cite{BalleW18}.

\subsection{Privacy analysis of \textsc{Policy Gaussian}}
 In this section we will prove that the \textsc{Policy Gaussian} algorithm (Algorithm~\ref{alg:policy-gaussian}) satifies $(\eps,\delta)$-DP. By Theorem~\ref{thm:sensitivity_implies_DP_informal} and Theorem~\ref{thm:contractivity_implies_sensitivity}, we already have an intuitive path to prove privacy. We now state privacy claims formally.
\begin{theorem}
\label{thm:policy-gaussian}
The \textsc{Policy Gaussian} algorithm (Algorithm \ref{alg:policy-gaussian}) is $(\epsilon, \delta)$-DP if $\sigma,\rho_\Gauss$ are chosen s.t.
\begin{align*}
&\Phi\left(\frac{1}{2\sigma}-\eps\sigma\right)-e^\eps\Phi\left(-\frac{1}{2\sigma}-\eps\sigma\right)\le \frac{\delta}{2} \text{   and }\\
&\rho_{\Gauss} \geq \max_{1\le t \le \Delta_0}\left(\frac{1}{\sqrt{t}}+\sigma \Phi^{-1}\left(\left(1-\frac{\delta}{2}\right)^{1/t}\right)\right).
\end{align*}
\end{theorem}
\begin{proof}
Suppose $D_1$ and $D_2$ are neighboring databases where $D_1$ has one extra user compared to $D_2$. Let $P$ and $Q$ denote the distribution of output of the algorithm when the database is $D_1$ and $D_2$ respectively.
We want to show that $P\approx_{\eps,\delta} Q$. It is enough to prove this for any fixed choice of $W_i'\subset W_i$ (in Algorithm~\ref{alg:meta_histogram}) identical in both instances, which corresponds to a coupling. Let $E$ be the event that $A\subset \supp(H_2).$ 
\begin{claim}
	$P|_E \approx_{\eps,\delta/2} Q$
\end{claim}

\begin{proof}
	Let $H_1$ and $H_2$ be the histograms generated by the algorithm from databases $D_1$ and $D_2$ respectively. And $\hat H_1$ and $\hat H_2$ be the histograms obtained by adding $\cN(0,\sigma^2)$ noise to each entry of $H_1$ and $H_2$ respectively. 
	By the post-processing lemma (Lemma~\ref{lem:post_processing_eps_delta}), it is enough to show that the distributions of the $\hat H_1|_{\supp(H_2)}$ and $\hat H_2$ are $(\eps,\delta/2)$-close to each other. Because the histogram building algorithm (Algorithm~\ref{alg:meta_histogram}) has $\ell_2$-sensitivity of at most 1 by Theorem~\ref{thm:contractivity_implies_sensitivity}, $\ellnorm{H_1|_{\supp(H_2)}-H_2}{2}\le 1$. Therefore by properties of Gaussian mechanism (Proposition~\ref{lem:gaussian_mechanism}), it is enough to choose $\sigma$ as in the statement of the theorem.
\end{proof}
By Lemma~\ref{lem:eps_delta_conditioning}, it is enough to show that $P(E)\ge 1-\delta/2$. Let $T=\supp(H_1)\setminus \supp(H_2).$ Note that $|T|\le \Delta_0$ and $H_1[u]\le \frac{1}{\sqrt{|T|}}$ for $u\in T$ by symmetry of the update policy.

\begin{align}
    P(\bar E) &= \Pr[\exists u \in T\ |\ \hat{H_1}[u] > \rho_{\Gauss}] \nonumber \\
    &= 1-\Pr[\forall u\in T\ \ \hat{H_1}[u] \le \rho_{\Gauss}]\nonumber \\
    &= 1-\prod_{u\in T} \Pr[\hat{H_1}[u] \le \rho_{\Gauss}]\nonumber \\
    &= 1-\prod_{u\in T} \Pr[H_1[u] + X_u \le \rho_{\Gauss}] \tag{$X_u \sim \cN(0,\sigma^2)$} \nonumber \\
    &\le 1-\prod_{u\in T} \Pr\left[X_u \le \rho_{\Gauss} - \frac{1}{\sqrt{|T|}}\right]  \tag{$H_1[u]\le \frac{1}{\sqrt{|T|}}$ for $u\in T$}\nonumber \\
    &= 1-\Phi\left(\frac{\rho_\Gauss}{\sigma}-\frac{1}{\sqrt{|T|}}\right)^{|T|}  \label{eqn:union_bound_gaussian}
\end{align}
Thus for $$\rho_{\Gauss} \geq \max_{1\le t \le \Delta_0}\left(\frac{1}{\sqrt{t}}+\sigma \Phi^{-1}\left(\left(1-\frac{\delta}{2}\right)^{1/t}\right)\right),$$ we have $P(\bar E)\le \delta/2$. 
Therefore the DP Set Union algorithm (Algorithm \ref{alg:meta}) is $(\eps,\delta)$-DP.
\end{proof}

\subsection{$\ell_2$-contractive update policies}
We will set some \emph{cutoff} $\Gamma$ above the threshold $\rho$ and once an item's count ($H[u]$) crosses the cutoff, we don't want to increase it further. In this policy, each user starts with a budget of 1. But now, the total change a user can make to the histogram can be at most $1$ when measured in $\ell_2$-norm. In other words, sum of the squares of the changes that the user makes is at most 1. The user wants the weights of items in their subset to get as close to the cutoff ($\Gamma$) as possible, say in some $\ell_q$-norm. Therefore the user moves the weights vector (restricted to the set $W$ of $\Delta_0$ items the user has) by an $\ell_2$-distance of at most 1 so as to get as close to the point $(\Gamma,\Gamma,\dots,\Gamma)$ as possible in $\ell_q$-norm. This is called \textsc{$\ell_q$-descent}. This update policy is presented in Algorithm~\ref{alg:ellq-descent}.

\begin{algorithm}[H]
\caption{\textsc{$\ell_q$-descent} update policy}
\label{alg:ellq-descent}
\begin{algorithmic}
	\STATE {\bfseries Input:} $H_0$: Current histogram \\
	$W$: A subset of $U$ of size at most $\Delta_0$\\
	$\Gamma$: cutoff parameter
	\STATE {\bfseries Output:} $H_1$: Updated histogram
	\STATE
	\STATE $H_1|_{U\setminus W} \leftarrow H_0|_{U\setminus W}$
	\STATE $H_1|_W \leftarrow \mathrm{argmin}_{y\in \R^W} \norm{(\Gamma,\Gamma,\dots,\Gamma) - y}_{\ell_q}$ where $\norm{y - H_0|_W}_{\ell_2}\le 1$
\end{algorithmic}
\end{algorithm}

The most interesting choices for $q$ are $q=1$ and $q=2$. We will now show that both these choices lead to $\ell_2$-contractive update policies. $q=1$ is intuitively preferable, because it is well-known that $\ell_1$-norm minimization is \emph{sparsity-inducing}. Therefore, we expect that in \textsc{$\ell_1$-descent} the weights of many items reach the maximum value of $\Gamma$, and subsequently these items will be output by Algorithm~\ref{alg:meta_histogram} with high probability. That is $\ell_1$-norm minimization is a good proxy for maximizing the number of items output by Algorithm~\ref{alg:meta_histogram}. We will demonstrate this in our experiments (Section~\ref{sec:experiments}).

\subsubsection{\textsc{$\ell_2$-descent} update policy for $\ell_2$-contractivity}

 This policy is obtained by setting $q=2$ in Algorithm~\ref{alg:ellq-descent} and can be implemented efficiently as shown in Algorithm~\ref{alg:ell2-descent-ell2-contractive}. This policy can also be interpreted as \emph{gradient descent} to minimize the $\ell_2$-distance between the current weighted histogram and the point $(\Gamma,\Gamma,\dots,\Gamma)$, hence the name \textsc{$\ell_2$-descent}. Since the gradient vector is in the direction of the line joining the current point and $(\Gamma,\Gamma,\dots,\Gamma)$, the \textsc{$\ell_2$-descent} policy is moving the current histogram towards $(\Gamma,\Gamma,\dots,\Gamma)$ by an $\ell_2$-distance of at most 1.

\begin{algorithm}[H]
 \caption{\textsc{$\ell_2$-descent} update policy for $\ell_2$-contractivity}
 \label{alg:ell2-descent-ell2-contractive}
 \begin{algorithmic}
 \STATE {\bfseries Input:} $H_0$: Current histogram\\
 $W$: A subset of $U$ of size at most $\Delta_0$\\
 $\Gamma$: cutoff parameter
 \STATE {\bfseries Output:} $H_1$: Updated histogram
 \STATE
 \STATE $H_1|_{U\setminus W}\leftarrow H_0|_{U \setminus W}$
 \STATE $G \leftarrow (\Gamma,\Gamma,\dots,\Gamma)-H_0|_W$ 
 \STATE // $G$ is the vector joining $H_0|_W$ to $(\Gamma,\Gamma,\dots,\Gamma)$
 \STATE $H_1|_W \leftarrow H_0|_W + \frac{G}{\max\{\norm{G}_{\ell_2},1\}}$ 
 \STATE // If $\norm{G}_{\ell_2}\le 1$, then update $H_0|_W$ to $(\Gamma,\Gamma,\dots,\Gamma)$. Else, move $H_0|_W$ in the direction of $(\Gamma,\Gamma,\dots,\Gamma)$ by an $\ell_2$-distance of at most 1.
\end{algorithmic}
\end{algorithm}

We will need the following geometric lemma to prove $\ell_2$-contraction.
\begin{lemma}
	\label{lem:plane_geometry_fact}
	Let $A,B,C$ denote the vertices of a triangle in the Euclidean plane. If $|AB|> 1,$ let $B'$ be the point on the side $AB$ which is at a distance of $1$ from $B$ and if $|AB|\le 1,$ define $B'=A$. $C'$ is defined similarly. Then $|B'C'|\le |BC|.$
\end{lemma}
\begin{figure}[ht]
\centering
\includegraphics[scale=0.5]{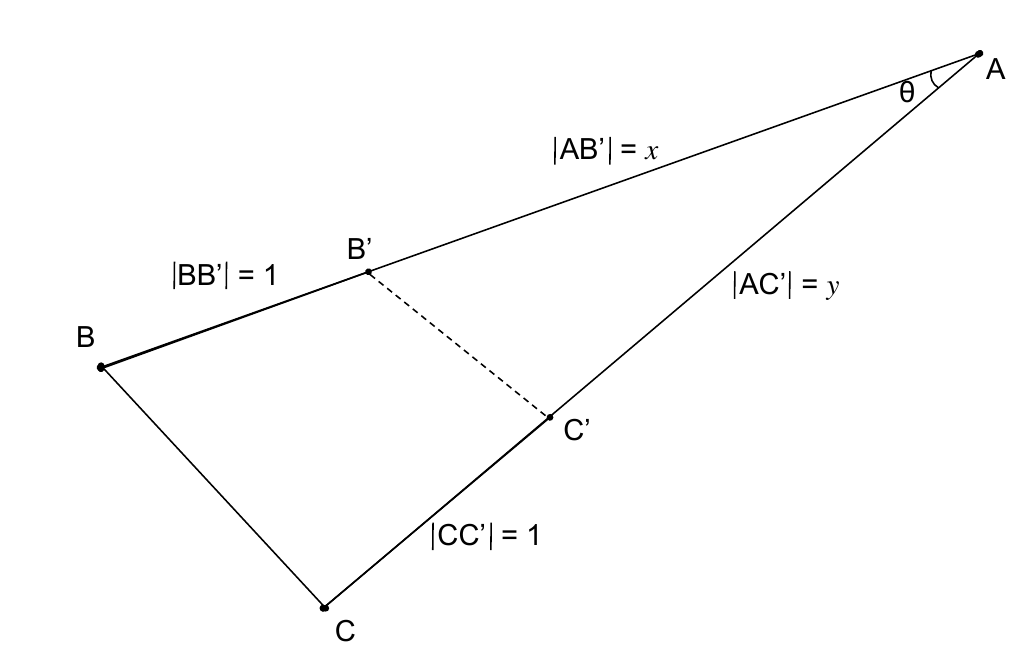}
\caption{\small Geometric explanation of Lemma \ref{lem:plane_geometry_fact} when $|AB|,|AC|> 1$. }
\label{fig:gaussian_geo_1}
\end{figure}
\begin{proof}
	Let us first assume that both $|AB|,|AC|>1.$
	Let $\theta$ be the angle at $A$ and let $|AB'|=x, |AC'|=y$ as shown in Figure~\ref{fig:gaussian_geo_1}. Then by the cosine formula,
	\begin{align*}
		|BC|^2 &= |AB|^2+|AC|^2-2|AB||AC|\cos\theta\\ 
		&= (x+1)^2+(y+1)^2-2(x+1)(y+1)\cos\theta\\
		&= x^2 + y^2 + 2xy\cos\theta + 2(x+y+1)(1-\cos\theta)\\
		&\ge x^2 + y^2+ 2xy\cos\theta \tag{$\cos \theta \le 1$}\\
		&=|B'C'|^2.
	\end{align*}
	\begin{figure}[ht]
	\centering
	\includegraphics[scale=0.5,trim=1.5cm 0.5cm 3cm 2cm, clip]{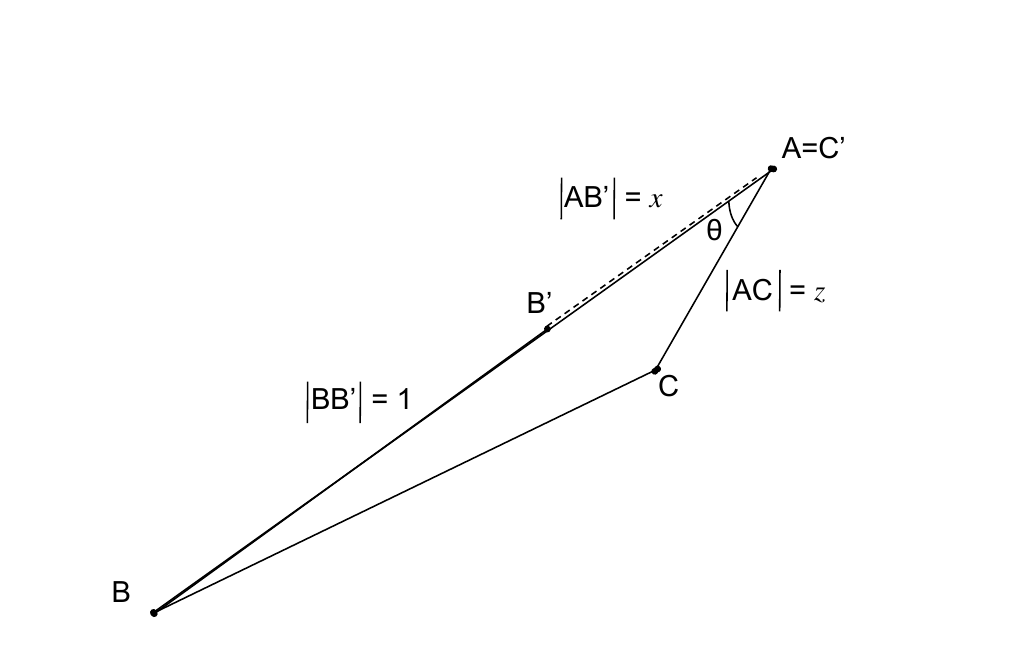}
	\caption{\small Geometric explanation of Lemma \ref{lem:plane_geometry_fact} when $|AB|>1,|AC|\le 1$. }
	\label{fig:gaussian_geo_2}
	\end{figure}
	If $|AB|,|AC|\le 1$, then $B'=C'=A$ and then the claim is trivially true. Suppose $|AB|>1,|AC|\le 1$. Now $C'=A.$ Let $|AB'|=x, |AC|=z\le 1$ and $\theta$ be the angle at $A$ as shown in Figure~\ref{fig:gaussian_geo_2}. Then by the cosine formula,
	\begin{align*}
		|BC|^2&= |AB|^2 + |AC|^2 - 2|AB||AC|\cos\theta\\
		&=(x+1)^2+z^2-2(x+1)z\cos\theta\\
		&=x^2+2x(1-z\cos\theta)+(z-\cos\theta)^2+(1-\cos^2\theta)\\
		&\ge x^2 = |AB'|^2=|B'C'|^2. \tag{$0\le z\le 1,|\cos\theta|\le 1$}
	\end{align*}
	By symmetry, the claim is also true when $|AC|>1,|AB|\le 1$.
\end{proof}

\begin{lemma}
	\label{lem:contraction_ell2}
	The \textsc{$\ell_2$-descent} policy in Algorithm~\ref{alg:ell2-descent-ell2-contractive} is $\ell_2$-contractive.
\end{lemma}
\begin{proof}
Suppose there are two histograms $G_1,G_2$ over some domain $X$. Suppose a $G_1,G_2$ are updated by a new user using the policy in Algorithm~\ref{alg:ell2-descent-ell2-contractive}. Let $G_1',G_2'$ be the updated histograms. Then we want to claim that $\norm{G_1'-G_2'}_{\ell_2}\le \norm{G_1-G_2}_{\ell_2}$.

	A new user can increase $G_1$ and $G_2$ by at most 1 in $\ell_2$ distance. Let $\Gamma$ be the cutoff parameter in Algorithm~\ref{alg:ell2-descent-ell2-contractive}. Let $W$ be the set of $\Delta_0$ items with the new user, therefore only the items in $W$ will change in $G_1',G_2'$. Therefore we can just assume that $G_1,G_2$ are supported on $W$ for the sake of the analysis. Algorithm~\ref{alg:ell2-descent-ell2-contractive} moves $G_i$ towards $P=(\Gamma,\Gamma,\dots,\Gamma)$ by an $\ell_2$-distance of 1 (or to $P$ if the distance to $P$ is already lower than 1). We can restrict ourselves to the plane containing $G_1,G_2,P$ ($G_1',G_2'$ will also lie on the same plane). Now by Lemma~\ref{lem:plane_geometry_fact}, $\norm{G_1'-G_2'}_{\ell_2}\le \norm{G_1-G_2}_{\ell_2}$.
\end{proof}

\subsubsection{\textsc{$\ell_1$-descent} update policy for $\ell_2$-contractivity}
\label{sec:maxsum}

 This policy is obtained by setting $q=1$ in Algorithm~\ref{alg:ellq-descent} and can be implemented efficiently as shown in Algorithm~\ref{alg:ell1-descent-ell2-contractive}. We will set some \emph{cutoff} $\Gamma$ above the threshold $\rho$ to use in the update policy. Once the weight of an item ($H[u]$) crosses the cutoff, we do not want to increase it further. In this policy, each user starts with a budget of 1 (measured in $\ell_2$ norm). The user uniformly increases $H[u]$ for each $u\in W$ s.t. $H[u]<\Gamma$. Once some item's weight reaches $\Gamma,$ the user stops increasing that item and keeps increasing the rest of the items uniformly until the budget of 1 is expended.

\begin{algorithm}[H]
 \caption{\textsc{$\ell_1$-descent} update policy for $\ell_2$-contractivity}
 \label{alg:ell1-descent-ell2-contractive}
 \begin{algorithmic}
 \STATE {\bfseries Input:} $H_0$: Current histogram \\
 $W$: A subset of $U$ of size at most $\Delta_0$\\
$\Gamma$: cutoff parameter
\STATE {\bfseries Output:} $H_1$: Updated histogram
\STATE
\STATE $H_1|_{U\setminus W} \leftarrow H_0|_{U\setminus W}$
\STATE $G\leftarrow (\Gamma,\Gamma,\dots,\Gamma)-H_0|_W$ 
\IF {$\norm{G}_{\ell_2} \le 1$}
	\STATE $H_1|_W \leftarrow H_0|_W$
\ELSE
	\STATE Find $\lambda \ge 0$ such that $\sum_{u\in W} \min\{G[u],\lambda\}^2 = 1$
	\STATE $H_1[u] \leftarrow H_0[u] + \min\{G[u],\lambda\} \ \forall u\in W$
\ENDIF
\end{algorithmic}
\end{algorithm}

\begin{proposition}
	The \textsc{$\ell_1$-descent} policy in Algorithm~\ref{alg:ell1-descent-ell2-contractive} is $\ell_2$-contractive.
\end{proposition}
\begin{proof}
	Fix some cutoff $\Gamma$. We can ignore what is happening outside the set $W$, let $d=|W|$. Given a histogram $x\in \R^d$, let $F(x)\in \R^d$ be the updated histogram according to the \textsc{$\ell_1$-descent} policy in Algorithm~\ref{alg:ell1-descent-ell2-contractive}. We have
	\begin{align*}
		F(x) =&\ \underset{y}{\mathrm{argmax}} \sum_{i=1}^d y_i\\
			& \text{s.t. } y_i \le \Gamma\ \ \forall i \text{ and } \norm{y-x}_{\ell_2}\le 1. 
	\end{align*}
	We want to prove that $\norm{F(x)-F(x')}_{\ell_2}\le \norm{x-x'}_{\ell_2}.$ Note that $F$ is continuous everywhere and differentiable almost everywhere. Therefore it is enough to show that the spectral norm of the Jacobian of $F$, $\norm{J_F(x)}_{S_\infty} \le 1$ whenever $F$ is differentiable at $x$. Fix such an $x$ and WLOG assume that $\Gamma>x_1>x_2>\dots>x_d.$ Note that $F(x)$ will have the form $(\Gamma,\dots,\Gamma, <\Gamma,\dots,<\Gamma)$. Let $F(x)_i=\Gamma$ for $i\in [t]$ and $F(x)<\Gamma$ for $i>t.$ Now we can explicitly compute $F(x)$ as:

	\begin{align*}
	 	F(x)_i = \begin{cases}
	 		\Gamma &\text{ if } i\in [t]\\
	 		x_i + \lambda & \text{ if } i>t
	 	\end{cases}
	 \end{align*}
where $\lambda\ge \Gamma-x_t$ is such that $(\Gamma-x_1)^2+\dots+(\Gamma-x_t)^2+(d-t)\lambda^2 =1.$ Note that $\lambda$ therefore only depends on $x_1,\dots,x_t.$ Therefore the Jacobian has the block diagonal form:
\begin{align*}
	J_F(x) = \left[\frac{\partial F_i}{\partial x_j}\right] = 
	\left[\begin{array}{c|c}
	\bigzero & \bigzero\\
	\hline
	* & I
	\end{array} \right].
\end{align*}
Therefore $\norm{J_F(x)}_{S_\infty}\le 1.$
\end{proof}

\section{Experiments}
\label{sec:experiments}
While the algorithms we described generalize to many domains that involve the release of set union, our experiments will use a natural language dataset. In the context of n-gram release, $D$ is a database of users where each user is associated with 1 or more Reddit posts and $W_i$ is the set of unique n-grams used by each user. The goal is to output as large a subset of $n$-grams $\cup_i W_i$ as possible while providing $(\eps,\delta)$-differential privacy to each user. In our experiments we consider $n=1, 2, 3$ (i.e. unigrams, bigrams, and trigrams)\footnote{The code and dataset used for our experiments are available at \url{https://github.com/heyyjudes/differentially-private-set-union}}.

\subsection{Dataset}
\begin{figure}[ht]
\begin{center}
\includegraphics[width=8cm]{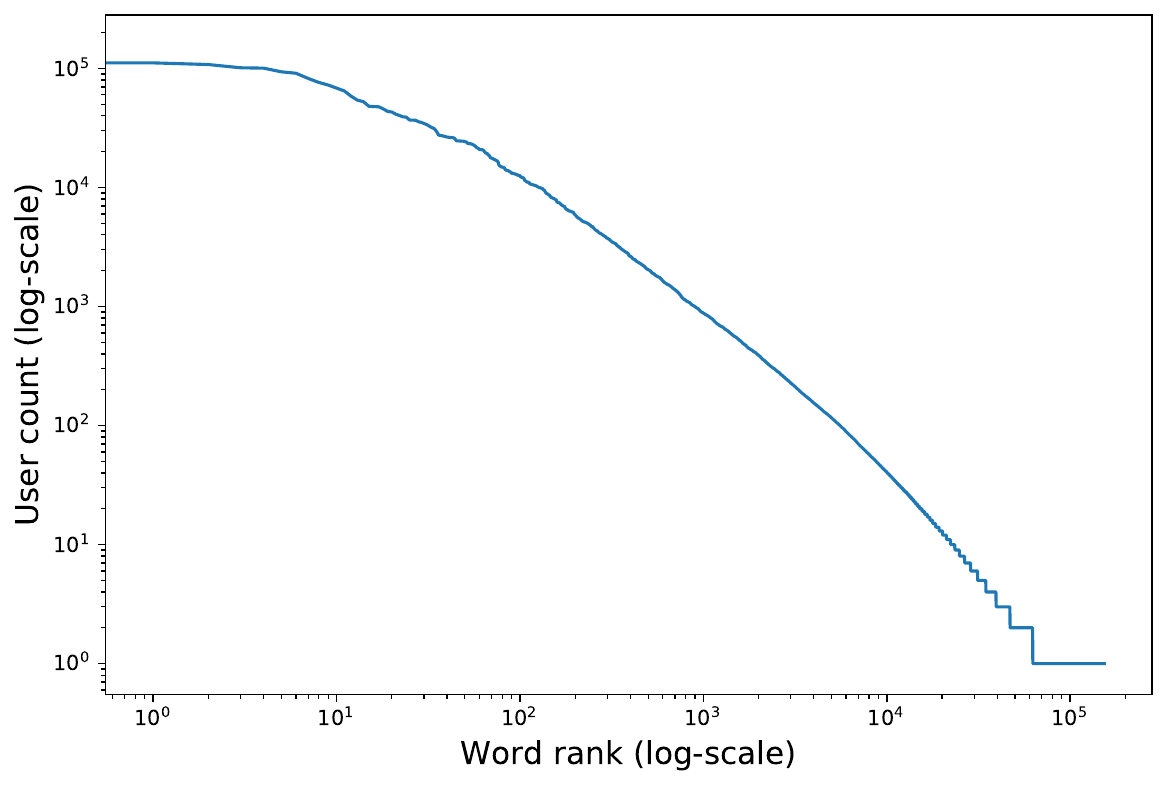}
\caption{\small Frequency (i.e. number of users who use the unigram) vs. rank of the unigram (based on frequency) on a log-log scale. This linear relationship shows that the frequency of unigrams among users also follows Zipf's law (power law), i.e., $\mathrm{count}\propto 1/\mathrm{rank}^\alpha$ for some constant $\alpha>0$. The $\alpha$ in this case is $\approx 1$.}
\label{fig:zipf_law}
\end{center}
\end{figure}
Our dataset is collected from the subreddit \texttt{r/AskReddit}. We take a sample of 15,000 posts from each month between January 2017 and December 2018. We filter out duplicate entries, removed posts, and deleted authors. For text preprocessing, we remove URLs and symbols, lowercase all words, and tokenize using \texttt{nltk.word\_tokenize}. After preprocessing, we again filter out empty posts to arrive at a dataset of 373,983 posts from 223,388 users. 

Similar to other natural language datasets, this corpus follows Zipf's law across users. The frequency of unigrams across users is inversely proportional to some power of the rank of the unigram. Using a log-log scale, the frequency of users for each unigram vs. the rank of the unigram is linear (Figure \ref{fig:zipf_law}). In other words, the lowest ranked (most common) unigrams are used by almost all users while the highest ranked (least common) unigrams are used by very few users. 

\begin{table}[ht]
\caption{\small Percentage of users with unique unigram count of less than or equal to $T$. The vast majority of users have less than 100 unique unigrams.}
\label{tab:user_words}
\begin{center}
\begin{small}
\begin{sc}
\begin{tabular}{lr}
\toprule
Threshold ($T$)   & Users with $|W_i| \leq$ T\\ \midrule
1   & 2.78\%                                   \\
10  & 29.82\%                               \\ 
50  & 79.16\%                                \\ 
100 & 93.13\%                               \\ 
300 & 99.59\%                               \\ \bottomrule
\end{tabular}
\end{sc}
\end{small}
\end{center}
\end{table}

The distribution of how many unigrams each user uses also follows a long tail distribution. While the top 10 users contribute between 850 and 2000 unique unigrams, most users (93.1\%) contribute less than 100 unique unigrams. Table \ref{tab:user_words} summarizes the percentage of users with a unique vocabulary smaller than each threshold T provided.

\subsection{Results}
\begin{table*}[ht]
\caption{\small Count of unigrams released by various set union algorithms. Results are averaged across 5 shuffles of user order. The best results for each algorithm are in bold. The privacy parameters are $\eps=3$ and $\delta=\exp(-10)$. The cutoff $\Gamma$ is calculated using $\alpha=3$ for \textsc{Policy Laplace} and \textsc{Policy Gaussian $\ell_2$} and $\alpha=5$ for all other algorithms} 
\label{tab:main_results}
\begin{center}
\begin{small}
\begin{sc}
\begin{tabular}{lrrrrr}
\toprule
$\Delta_0$  & \textbf{1}  & \textbf{10}  & \textbf{50} & \textbf{100} & \textbf{200}\\
\midrule
\textsc{Count Laplace}     & \textbf{4484} $\pm$ 32&  3666 $\pm$  7&  2199 $\pm$  8&  1502 $\pm$ 14 &   882 $\pm$  4\\
\textsc{Count Gaussian}    & 3179 $\pm$ 15&  6616 $\pm$ 18&  \textbf{6998} $\pm$ 23&  6470 $\pm$ 12 &  5492 $\pm$ 14\\
\textsc{Weighted Laplace}  & \textbf{4479} $\pm$ 26&  4309 $\pm$ 15&  4012 $\pm$ 10&  3875 $\pm$  9 &  3726 $\pm$ 17\\
\textsc{Weighted Gaussian} & 3194 $\pm$ 11&  6591 $\pm$ 18&  8570 $\pm$ 14&  8904 $\pm$ 24 &  \textbf{8996} $\pm$ 30\\
\textsc{Policy Laplace}    & 4387 $\pm$ 14& 12838 $\pm$ 42& \textbf{15421} $\pm$ 15& 14923 $\pm$ 2 & 14346 $\pm$ 24\\
\textsc{Policy Gaussian $\ell_1$}   & 3169 $\pm$ 13& 11010 $\pm$ 15& 16181 $\pm$ 33& 16954 $\pm$ 58 & \textbf{17113} $\pm$ 16\\
\textsc{Policy Gaussian $\ell_2$}   & 3180 $\pm$ 14& 10918 $\pm$ 24& 16188 $\pm$ 34& 17024 $\pm$ 16 & \textbf{17211} $\pm$ 37\\
\bottomrule
\end{tabular}
\end{sc}
\end{small}
\end{center}
\end{table*}

For the problem of outputting the large possible set of unigrams, Table \ref{tab:main_results} and Figure \ref{fig:main_table_fig}, summarize the performance of DP set union algorithms for different values of $\Delta_0$. The privacy parameters are $\eps=3$ and $\delta=\exp(-10)$. We compare our algorithms with baseline algorithms: \textsc{Count Laplace}, \textsc{Count Gaussian}, \textsc{Weighted Laplace}, and \textsc{Weighted Gaussian} discussed in Section~\ref{sec:baseline}. We use `\textsc{Policy Gaussian $\ell_2$}' to refer to \textsc{Policy Gaussian} algorithm which uses the \textsc{$\ell_2$-descent} update policy in Algorithm~\ref{alg:ell2-descent-ell2-contractive}. `\textsc{Policy Gaussian $\ell_1$}' refers to \textsc{Policy Gaussian} algorithm which uses the \textsc{$\ell_1$-descent} update policy in Algorithm~\ref{alg:ell1-descent-ell2-contractive}. Since we only present one $\ell_1$-contractive policy for \textsc{Policy Laplace} algorithm, in our experiments, \textsc{Policy Laplace} refers to the $\textsc{Policy Laplace}$ algorithm which uses the \textsc{$\ell_1$-descent} update policy in Algorithm~\ref{alg:ell1-descent-ell1-contractive}.

\begin{figure}[ht]
\begin{center}
\centerline{\includegraphics[width=12cm]{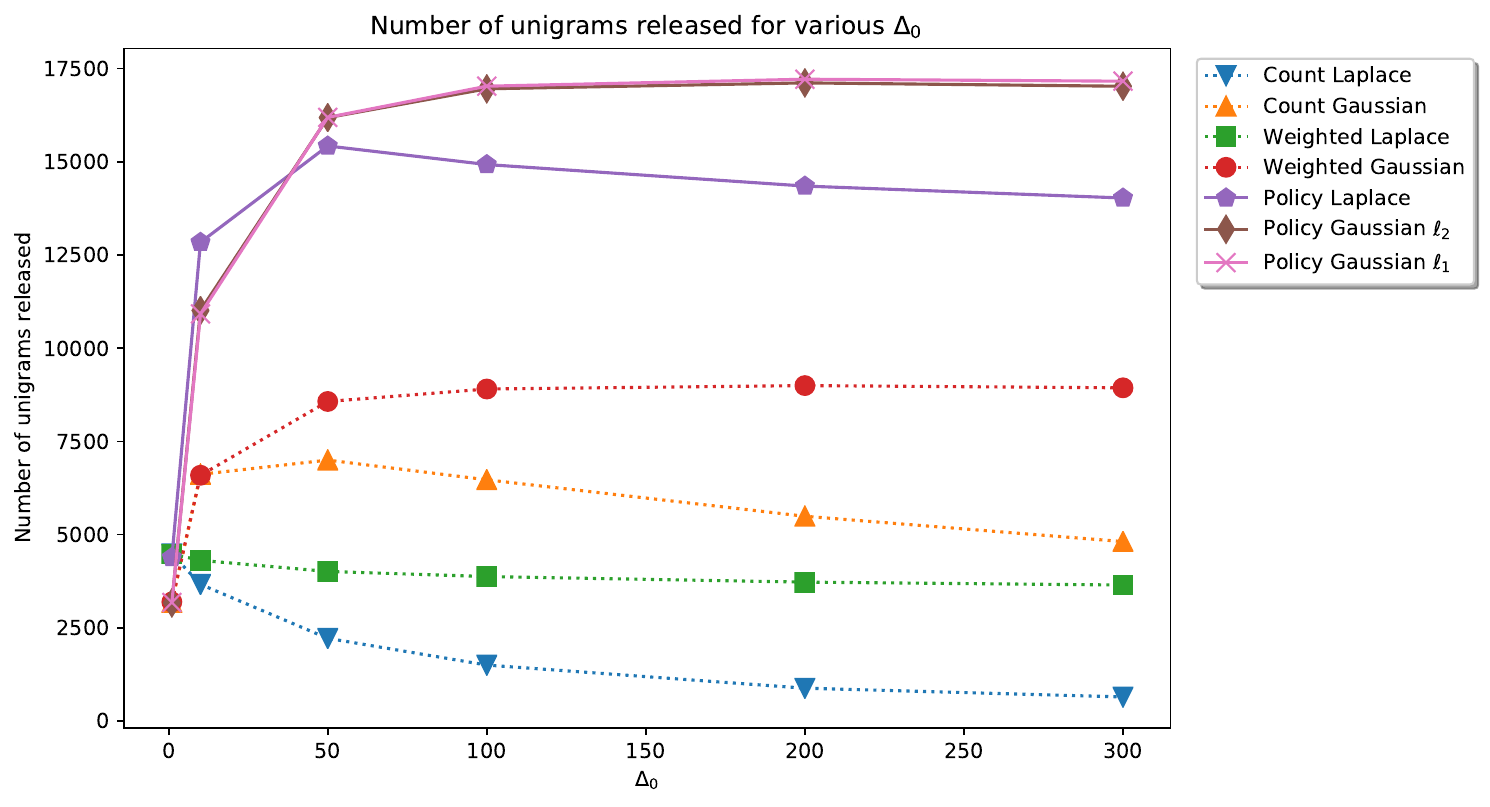}}
\caption{\small Count of unigrams released by set union algorithms averaged across 5 shuffles of user order. Privacy parameters are $\eps=3$ and $\delta=\exp(-10)$. The cutoff $\Gamma$ is calculated using $\alpha=3$ for \textsc{Policy Laplace} and \textsc{Policy Gaussian $\ell_1$} and $\alpha=5$ for \textsc{Policy Gaussian $\ell_2$}}
\label{fig:main_table_fig}
\end{center}
\end{figure}

\noindent Our conclusions are as follows:
\begin{itemize}
\item Our new algorithms \textsc{Policy Laplace} and \textsc{Policy Gaussian} output  a DP set union that is 2-4 times larger than output of  weighted/count based algorithms. This holds for all values of $\epsilon \geq 1$ (see Figure~\ref{fig:eps_fig1}).

\item To put the size of released set in context, we compare our new algorithms against the number of unigrams belonging to at least $k$ users (See Table \ref{tab:kanon}). For \textsc{Policy Laplace} with $\Delta_0=100$, the size of the output set covers almost all unigrams (94.8\%) when $k=20$ and surpasses the size of the output set when $k \geq 25$. \textsc{Policy Gaussian} with $\Delta_0=100$ covers almost all unigrams (91.8\%) when $k=15$ and surpasses the size of the output set when $k \geq 18$. In other words, our algorithms (with $\eps=3$ and $\delta=\exp(-10)$) outperform \textsc{$k$-anonymity} based algorithms for values of $k$ around 20.  
\end{itemize}

\begin{table}[ht]
\caption{\small This table shows the total number of unigrams that at least $k$ users possess ($|S_{k}|$) and the percentage coverage of this total by \textsc{Policy Laplace} ($|S_{PL}|=14739$) and \textsc{Policy Gaussian $\ell_2$} ($|S_{PG}|=16954$) for $\Delta_0=100$.}
\label{tab:kanon}
\begin{center}
\begin{small}
\begin{sc}
\begin{tabular}{lp{1cm}p{6cm}p{6cm}}
\toprule
$k$   & $|S_{k}|$  & \% coverage \textsc{Policy Laplace} & \% coverage \textsc{Policy Gaussian $\ell_2$} \\ \midrule
5  &  34699 & 24.5\% & 48.9\%  \\ 
10 &  23471 & 62.8\% & 72.2\%  \\
15 &  18461 & 79.8\% & 91.8\%  \\
18 &  16612 & 88.7\% & 102.1\%  \\
20 &  15550 & 94.8\% & 109.0\%  \\
25 &  13638 & 108.1\%& 124.3\%  \\ \bottomrule
\end{tabular}
\end{sc}
\end{small}
\end{center}
\end{table}

\subsubsection{Beyond Unigrams}
We also conduct experiments to compare the number of bigrams and trigrams released by each algorithm. This result is of interest when retrieving longer $n$-grams in real world settings. From Figure \ref{fig:bigram_delta} and \ref{fig:bigram_eps}, we see that the Gaussian mechanisms do better than Laplace algorithms across various values of $\Delta_0$ and $\eps$. While the \textsc{Policy Gaussian $\ell_1$} and \textsc{Policy Gaussian $\ell_2$} mechanisms performed similarly on unigrams, we see that \textsc{Policy Gaussian $\ell_1$} releases more bigrams across various parameter values. Looking at trigrams, we see that the total number output is smaller than bigrams. This is reasonable since the number of $n$-grams that only occur once increases as $n$ increases. Figure \ref{fig:trigram_delta} and Figure \ref{fig:trigram_eps} shows that the \textsc{Policy Gaussian $\ell_1$} and \textsc{Weighted Gaussian} perform best on trigrams. We expect this pattern to continue for $n>3$. 

\begin{figure}[ht]
    \centering
    \begin{minipage}{0.49\textwidth}
        \centering
        \includegraphics[width=\textwidth]{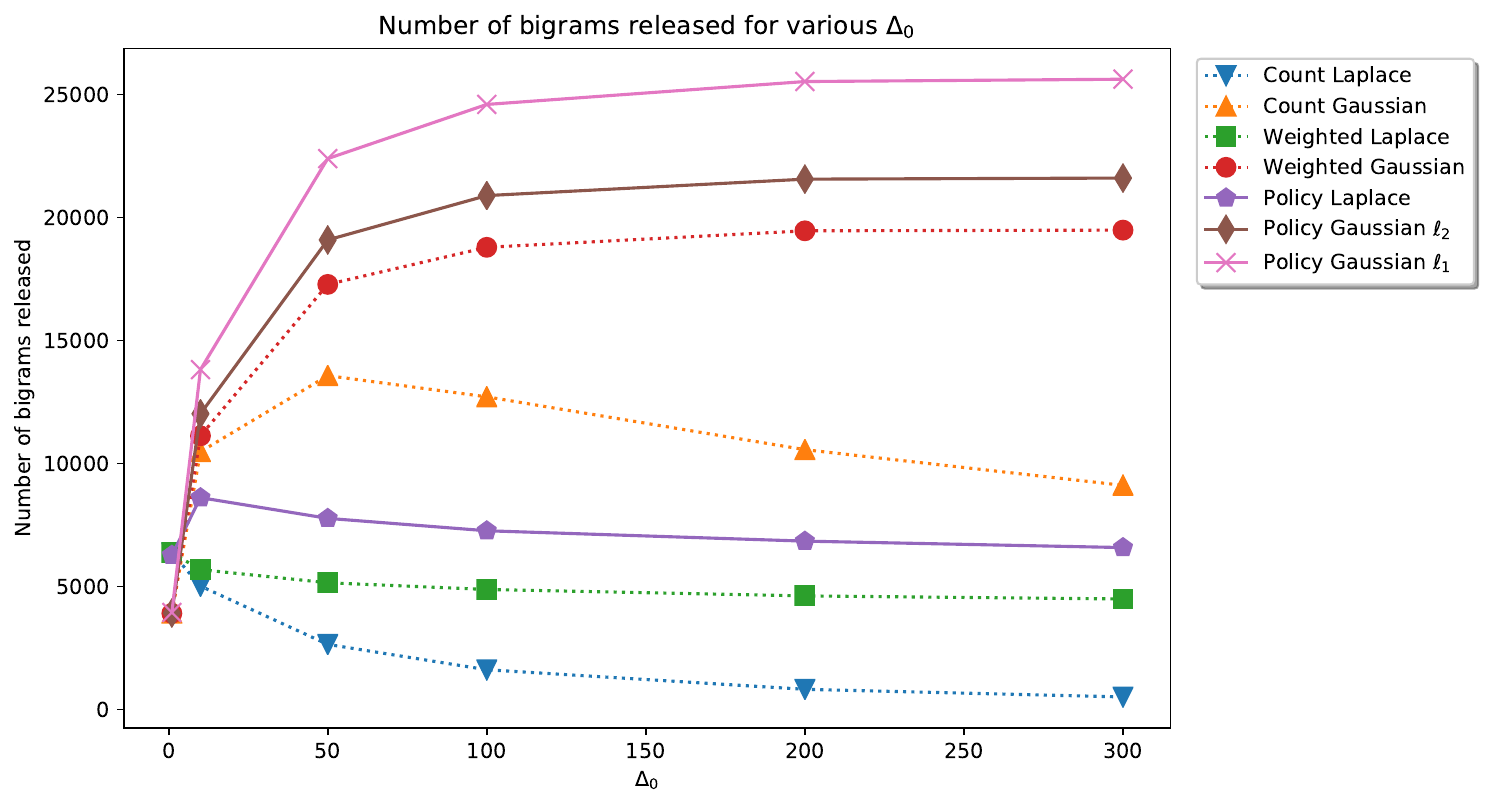} 
        \caption{\small Count of bigrams released at various values of $\Delta_0$ for parameters $\eps=3$, $\delta=\exp(-10)$.}
        \label{fig:bigram_delta}
    \end{minipage}\hfill
    \begin{minipage}{0.49\textwidth}
        \centering
        \includegraphics[width=\textwidth]{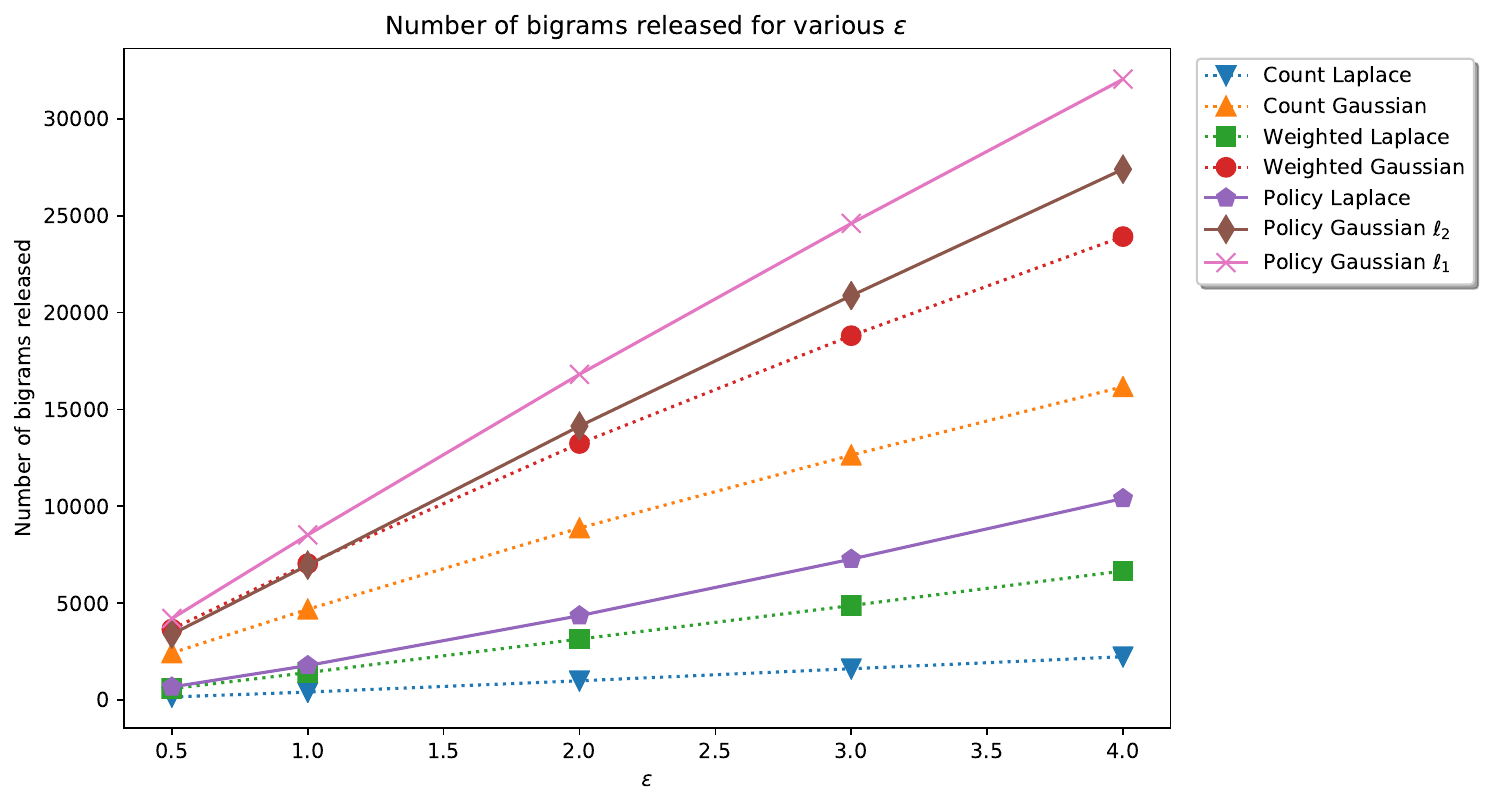} 
        \caption{\small Count of bigrams released at various values of $\eps$ for parameters $\Delta_0=100$, $\delta=\exp(-10)$.}
        \label{fig:bigram_eps}
    \end{minipage}
\end{figure}

\begin{figure}[ht]
    \centering
    \begin{minipage}{0.49\textwidth}
        \centering
        \includegraphics[width=\textwidth]{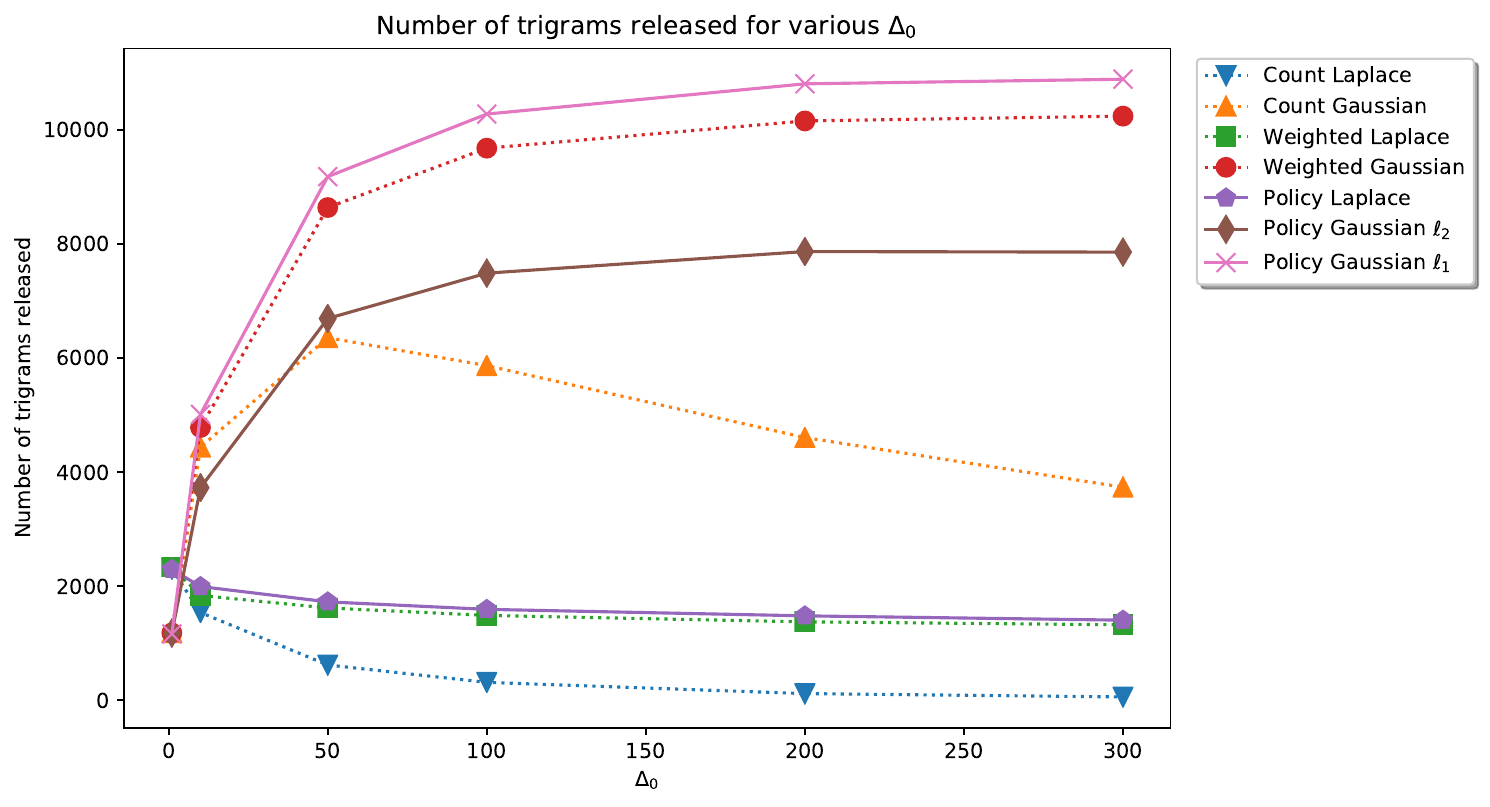} 
        \caption{\small Count of trigrams released at various values of $\Delta_0$ for parameters $\eps=3$, $\delta=\exp(-10)$.}
        \label{fig:trigram_delta}
    \end{minipage}\hfill
    \begin{minipage}{0.49\textwidth}
        \centering
        \includegraphics[width=\textwidth]{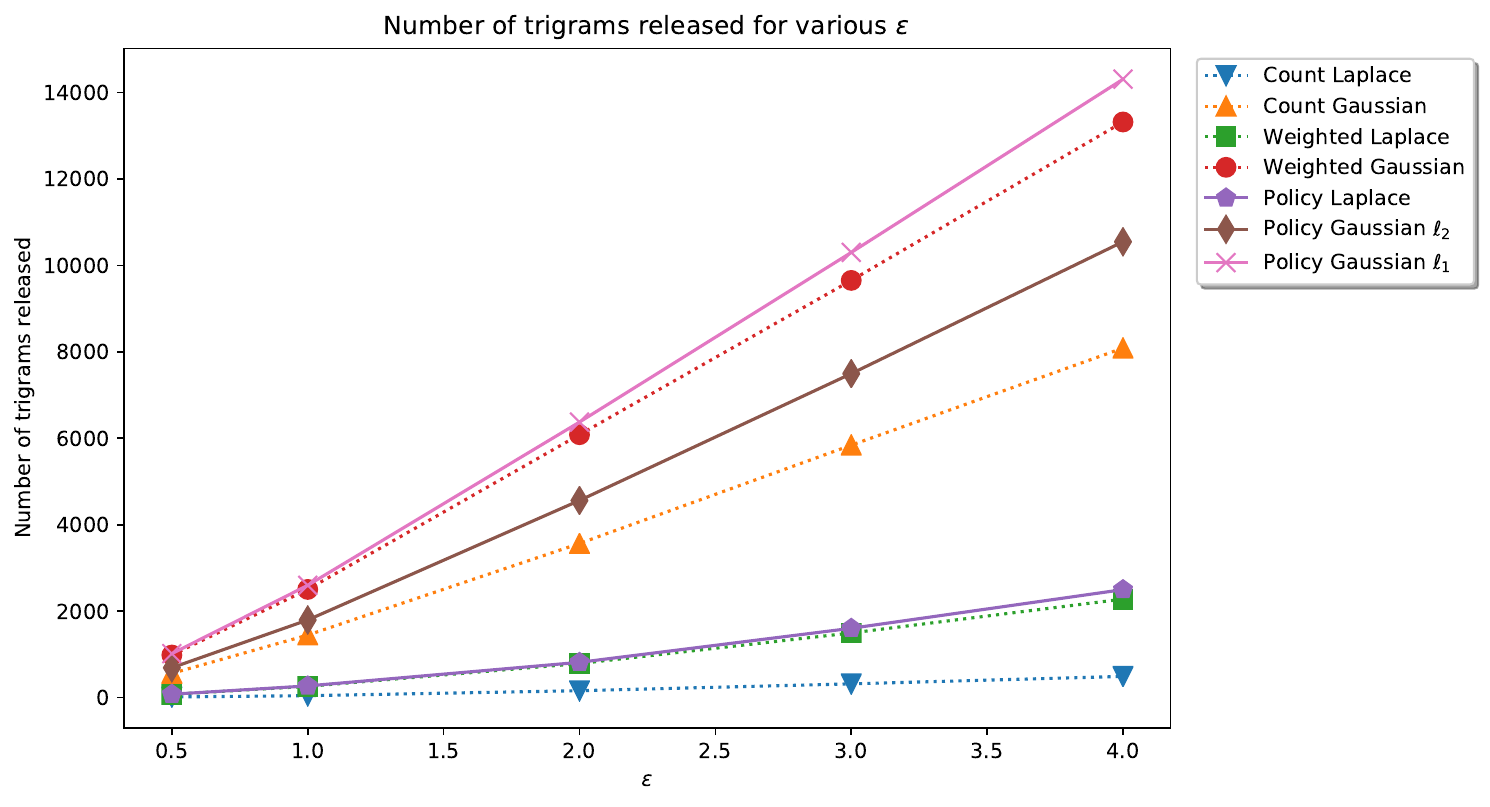} 
        \caption{\small Count of trigrams released at various values of $\eps$ for parameters $\Delta_0=100$, $\delta=\exp(-10)$.}
        \label{fig:trigram_eps}
    \end{minipage}
\end{figure}

\begin{figure}[ht]
    \centering
    \begin{minipage}{0.49\textwidth}
        \centering
        \includegraphics[width=\textwidth]{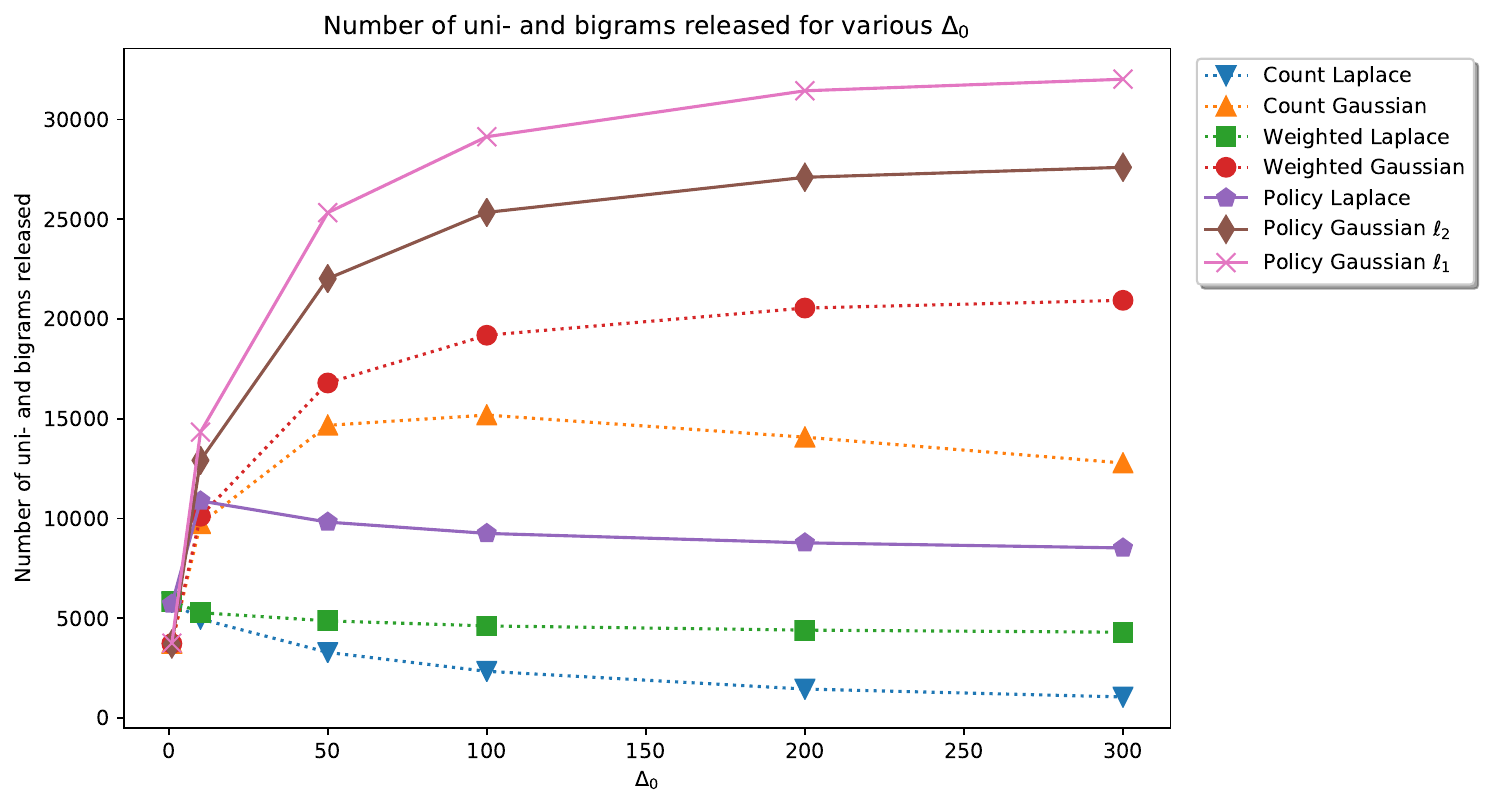} 
        \caption{\small Count of unigram and bigrams released at various values of $\Delta_0$ for parameters $\eps=3$, $\delta=\exp(-10)$.}
        \label{fig:bigram_union}
    \end{minipage}\hfill
    \begin{minipage}{0.49\textwidth}
        \centering
        \includegraphics[width=\textwidth]{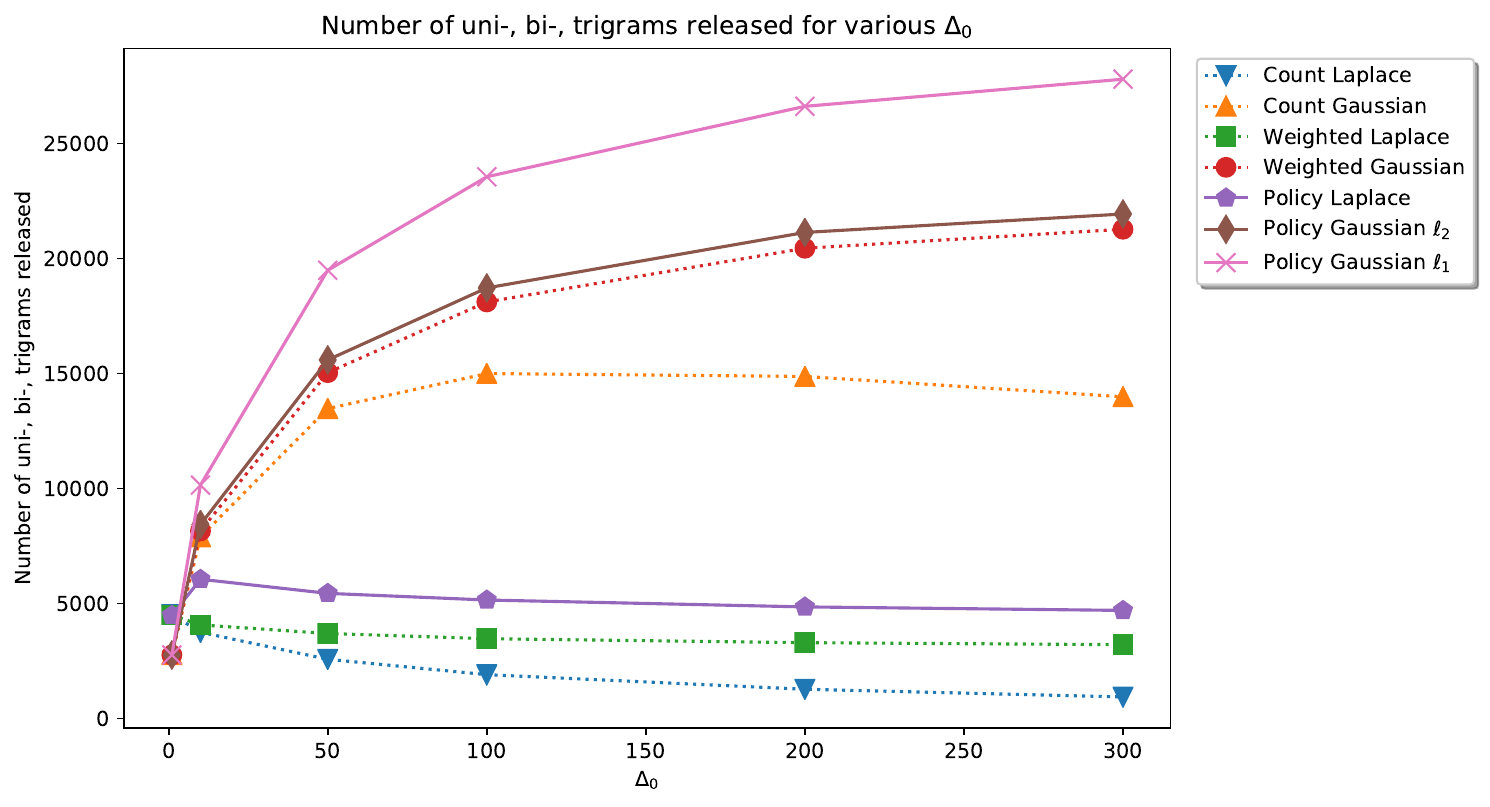} 
        \caption{\small Count of unigram, bigrams, and trigrams released at various values of $\Delta_0$ for parameters $\eps=3$, $\delta=\exp(-10)$.}
        \label{fig:trigram_union}
    \end{minipage}
\end{figure}

We also examine the case where we consider the union of $n$-grams for $n=1$ to $N$. When $N=1$, this is just the unigram case. When $N=2$ we consider items to be the union of all unigrams and bigrams. The results in Figure \ref{fig:bigram_union} and Figure \ref{fig:trigram_union} show the number of items output for $N=2$ (all unigrams and bigrams) and $N=3$ (all unigrams, bigrams, and trigrams). We see a similar pattern that all the Gaussian noise algorithms perform better than Laplace noise based algorithms. The \textsc{Policy Gaussian $\ell_1$} algorithm surpasses all other algorithms and emerges as the clear winner across all our different experiments.

\subsubsection{Multiple passes through each user}
\begin{table*}[ht]
\caption{\small Count of unigrams released by \textsc{Policy Laplace} and \textsc{Policy Gaussian $\ell_2$} algorithms for single and double passes over users. Results are averaged and rounded across 5 shuffles of user order. The privacy parameters are $\eps=3$ and $\delta=\exp(-10)$. $\alpha=2$ is chosen for the threshold parameter. Significant p-values for a two-sided independent $t$-test are in bold.} 
\label{tab:multi_iter_results}
\begin{center}
\begin{small}
\begin{sc}
\begin{tabular}{llll|lll}
\toprule
                           & \multicolumn{3}{l}{\textsc{Policy Laplace}}                       & \multicolumn{3}{l}{\textsc{Policy Gaussian $\ell_2$}}    \\ \midrule
$\Delta_0$ & 1 Pass                  & 2 Passes       & P-val          & 1 Pass         & 2 Passes       & P-val \\ \midrule
1                          & 4236 $\pm$ 14           & 4257 $\pm$ 17  & 0.083          & 3135 $\pm$ 25  & 3131 $\pm$ 20  & 0.829 \\ 
10                         & \textbf{12452 $\pm$ 31} & 12389 $\pm$ 17 & \textbf{0.008} & 10784 $\pm$ 22 & 10817 $\pm$ 54 & 0.293 \\ 
50                         & 15056 $\pm$ 35          & 15080 $\pm$ 21 & 0.262          & 15763 $\pm$ 33 & 15809 $\pm$ 45 & 0.139 \\ 
100                        & 14562 $\pm$ 50          & 14567 $\pm$ 24 & 0.846          & 14562 $\pm$ 50 & 14568 $\pm$ 24 & 0.846 \\ 
200                        & 14005 $\pm$ 33          & 13979 $\pm$ 31 & 0.271          & 14005 $\pm$ 33 & 13979 $\pm$ 31 & 0.271 \\ 
300                        & 13702 $\pm$ 37          & 13678 $\pm$ 47 & 0.448          & 13702 $\pm$ 37 & 13678 $\pm$ 47 & 0.447 \\ \bottomrule
\end{tabular}
\end{sc}
\end{small}
\end{center}
\end{table*}
In the experiments described thus far, each user contributes items once within the budget constraints. We also investigate whether the output of set union increases in size when each user contributes the same budget over multiple passes (e.g. user 1 contributes half of their budget each time over 2 passes), we compare \textsc{Policy Laplace} and \textsc{Policy Gaussian} outputs. Table \ref{tab:multi_iter_results} summarizes the results showing that there is not strong evidence suggesting that running multiple passes through the users improves the size of the output set.

\subsubsection{Selecting $\alpha$: parameter to set threshold $\Gamma$}
Figure \ref{fig:thresh_tune} shows the number of unigrams released by \textsc{Policy Laplace}, \textsc{Policy Gaussian $\ell_1$}, and \textsc{Policy Gaussian $\ell_2$} for various values of $\alpha$. We observe that the number of unigrams released increases sharply until $\alpha=3$, then remains nearly constant and then slowly decreases. We observe that \textsc{Policy Gaussian $\ell_1$} peaks at $\alpha=5$ while \textsc{Policy Laplace} and \textsc{Policy Gaussian $\ell_2$} peaks at $\alpha=3$. Thus we use these respective parameters for all of our experiments. This choice of $\alpha$ only affects the policy algorithms since the weighted and count algorithms do not use a threshold.
\begin{figure}[ht]
\begin{center}
\includegraphics[width=12cm]{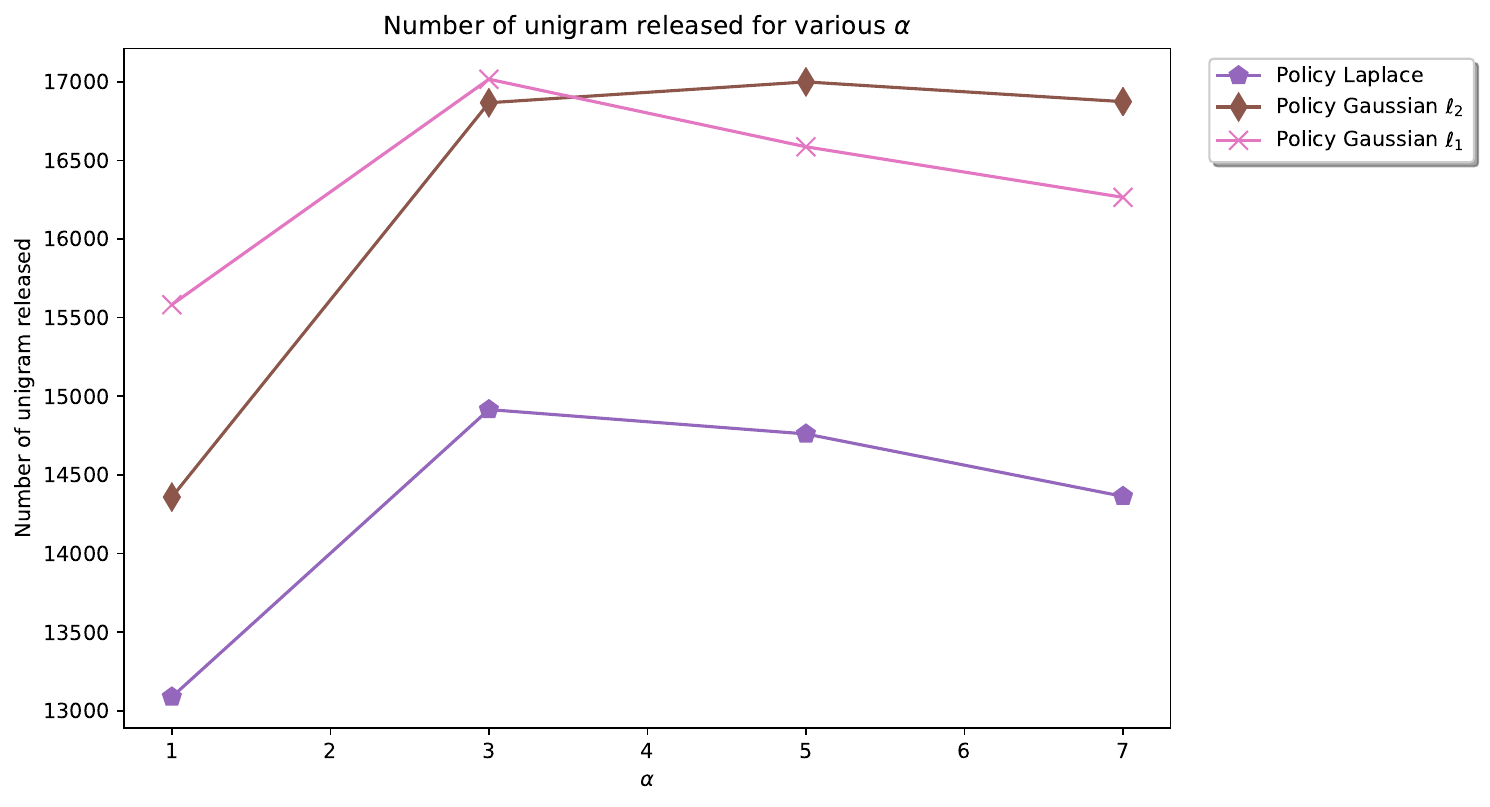}
\caption{\small Number of unigrams released for various values of $\alpha$ for \textsc{Policy Laplace}, \textsc{Policy Gaussian $\ell_1$}, and \textsc{Policy Gaussian $\ell_2$}. Here we fixed $\Delta_0=100$ and $\eps=3$.}
\label{fig:thresh_tune}
\end{center}
\end{figure}

\subsubsection{The effect of $\eps$}

We use $\eps = 3$ for the experiments in Table \ref{tab:main_results}. At this value of $\eps$ our policy algorithms perform much better than previous count and weighted algorithms. To check whether this result holds with smaller $\eps$, we also run these algorithms on various values of $\eps$. Figure \ref{fig:eps_fig1} shows that for $\eps \geq 1$ our policy algorithms always perform better for unigrams. For bigrams and trigrams, Gaussian noise algorithms perform best at various epsilons. Figure \ref{fig:bigram_eps} and Figure \ref{fig:trigram_eps} show that our \textsc{Policy Gaussian $\ell_1$} and \textsc{Policy Gaussian $\ell_2$} outperform \textsc{Count Gaussian} and \textsc{Weighted Gaussian} for $\eps \geq 2$.

\subsubsection{Selecting hyperparameters while maintaining privacy}
As can be seen from Table \ref{tab:main_results}  the $\Delta_0$ resulting in the largest output set varies by algorithm. Since most users in our dataset possess less than 300 unique unigrams, it is not surprising that the largest output set can be achieved with $\Delta_0 < 300$.  However, running our algorithms for different values of $\Delta_0$ and selecting the best output will result in a higher value of $\epsilon$. There are several ways to find the best value of $\Delta_0$ (or any other tunable parameter): 1) using prior knowledge of the data 2) running the algorithms on a small sample of the data to find the best parameters, and discarding that sample.
3) finally, one could also run all the algorithms in parallel and choose the best performing one. Here we will have to account for the loss in privacy budget; see \cite{LiuTalwar} for example.

\section{Conclusions and open problems}
We initiated the study of \emph{differentially private set union} (DPSU), which has many real-world applications.
We designed better algorithms for this problem using the notion of `contractive update policy' as a guiding principle. In our experiments, we demonstrated that our algorithms significantly outperform previous state-of-the-art algorithms. Algorithms based on Gaussian noise such as \textsc{Weighted Gaussian} and \textsc{Policy Gaussian} do much better than algorithms based on Laplace noise. In particular, the \textsc{Policy Gaussian} algorithm with \textsc{$\ell_1$-descent} update policy emerges as a clear winner across a range of scenarios. 

It would be interesting to find other contractive update policies which perform better than those we present in this paper. Another important open question is to explore how to parallelize our algorithms to enable them in scenarios where the input data is enormous and distributed across many machines. The policy based algorithms we introduce in this paper are harder to parallelize than algorithms like \textsc{Weighted Gaussian}. One possibility is to consider a hybrid approach, where each machine uses a policy based approach to update weights and the weights across machines are aggregated naively.


\bibliography{main}
\bibliographystyle{abbrvnat}
\newpage 
\appendix
\section{Bounded Sensitivity implies DP (Proof of Theorem~\ref{thm:sensitivity_implies_DP_informal})}
\label{sec:sensitivity_implies_DP}
In this section, we will prove a formal version of Theorem~\ref{thm:sensitivity_implies_DP_informal}, i.e., if the histogram output by Algorithm~\ref{alg:meta_histogram} has bounded $\ell_p$-sensitivity (for $p\in \{1,2\}$), then by adding appropriate noise and setting an appropriate threshold, Algorithm~\ref{alg:meta} for DP set union can be made differentially private. Here we do not assume that the contractive update policy is symmetric unlike in Theorems~\ref{thm:policy-laplace} and \ref{thm:policy-gaussian}. The lower bounds on the threshold ($\rho$) that we obtain in this generality are only slightly worse compared to the corresponding bounds in Theorems~\ref{thm:policy-laplace} and \ref{thm:policy-gaussian}.

\begin{theorem}
\label{thm:sensitivity_implies_DP_Laplace}
Suppose the histogram output by Algorithm~\ref{alg:meta_histogram} has $\ell_1$-sensitivity 1. Then Algorithm~\ref{alg:meta} is $(\epsilon, \delta)$-$DP$ when the \textsf{Noise} distribution is $\Lap(0,\lambda)$ where $\lambda=1/\eps$ and the threshold $$\rho \geq \max_{1\le t\le \Delta_0} 1+\frac{1}{\epsilon}\log\left(\frac{1}{2\left(1-(1-\delta)^{1/t}\right)}\right).$$
\end{theorem}
\begin{proof}
Proof of Theorem~\ref{thm:sensitivity_implies_DP_Laplace} is extremely similar to the proof of Theorem~\ref{thm:policy-laplace}. The only place where it differs is in Equation~(\ref{eqn:union_bound_laplace}) where we bound $H_1[u]\le 1$ instead of $H_1[u]\le 1/|T|.$
\end{proof}

\begin{theorem}
\label{thm:sensitivity_implies_DP_Gaussian}
Suppose the histogram output by Algorithm~\ref{alg:meta_histogram} has $\ell_2$-sensitivity 1. Then Algorithm~\ref{alg:meta} is $(\epsilon, \delta)$-$DP$ when the \textsf{Noise} distribution is $\cN(0,\sigma^2)$ where $\sigma$ and the threshold $\rho$ are chosen s.t.
\begin{align*}
&\Phi\left(\frac{1}{2\sigma}-\eps\sigma\right)-e^\eps\Phi\left(-\frac{1}{2\sigma}-\eps\sigma\right)\le \frac{\delta}{2} \text{ and}\\
&\rho \geq \max_{1\le t \le \Delta_0}\left(1+\sigma \Phi^{-1}\left(\left(1-\frac{\delta}{2}\right)^{1/t}\right)\right).
\end{align*}
\end{theorem}
\begin{proof}
Proof of Theorem~\ref{thm:sensitivity_implies_DP_Gaussian} is extremely similar to the proof of Theorem~\ref{thm:policy-gaussian}. The only place where it differs is in Equation~(\ref{eqn:union_bound_gaussian}) where we bound $H_1[u]\le 1$ instead of $H_1[u]\le 1/\sqrt{|T|}.$
\end{proof}

\section{Weighted Laplace and Gaussian algorithms}
\subsection{Weighted Laplace}

\begin{algorithm}[!h]
 \caption{\textsc{Laplace} weighted update}
 \label{alg:weighted-laplace}
 \begin{algorithmic}
 \STATE {\bfseries Input:} $H$: Current histogram \\
 $W$: A subset of $U$ of size at most $\Delta_0$
\STATE {\bfseries Output:} $H$: Updated histogram
 \FOR{$u$ in $W$}
	\STATE H[$u$] $\leftarrow$ H[$u$] + $\frac{1}{|W|}$
\ENDFOR
	\end{algorithmic}
\end{algorithm}

\begin{theorem}
\label{thm:weighted-laplace}
The \textsc{Weighted Laplace} algorithm (Algorithm \ref{alg:weighted-laplace}) is $(\epsilon, \delta)$-$DP$ when $$\rho_{\Lap} \geq \max_{1\le t\le \Delta_0} \frac{1}{t}+\frac{1}{\epsilon}\log\left(\frac{1}{2\left(1-(1-\delta)^{1/t}\right)}\right).$$
\end{theorem}
\begin{proof} Proof is exactly the same as that of Theorem~\ref{thm:policy-laplace}.
\end{proof}

\subsection{Weighted Gaussian}
\begin{algorithm}[!h]
 \caption{\textsc{Gaussian} weighted update}
 \label{alg:weighted-gaussian}
 \begin{algorithmic}
 \STATE {\bfseries Input:} $H$: Current histogram \\
 $W$: A subset of $U$ of size at most $\Delta_0$
\STATE {\bfseries Output:} $H$: Updated histogram
 \FOR{$u$ in $W$}
	\STATE H[$u$] $\leftarrow$ H[$u$] + $\sqrt{\frac{1}{|W|}}$
\ENDFOR
	\end{algorithmic}
\end{algorithm}

\begin{theorem}
\label{thm:weighted-gaussian}
The \textsc{Weighted Gaussian} algorithm (Algorithm \ref{alg:weighted-gaussian}) is $(\epsilon, \delta)$-DP if $\sigma,\rho_\Gauss$ are chosen s.t.
\begin{align*}
&\Phi\left(\frac{1}{2\sigma}-\eps\sigma\right)-e^\eps\Phi\left(-\frac{1}{2\sigma}-\eps\sigma\right)\le \frac{\delta}{2} \text{   and   }\\
&\rho_{\Gauss} \geq \max_{1\le t \le \Delta_0}\left(\frac{1}{\sqrt{t}}+\sigma \Phi^{-1}\left(\left(1-\frac{\delta}{2}\right)^{1/t}\right)\right).
\end{align*}
\end{theorem}
\begin{proof} Proof is exactly the same as that of Theorem~\ref{thm:policy-gaussian}.
\end{proof}

\end{document}